\documentclass{article}[11pt]  
\usepackage{amsmath}
\usepackage{amssymb}
\usepackage{amsthm}
\usepackage{ifpdf}
\usepackage{wrapfig}
\usepackage{fullpage}
\usepackage{cite}
\usepackage{subcaption}
\usepackage{float}
\usepackage{graphicx}

\ifpdf

  \usepackage[pdftex]{epsfig}
  \usepackage[pdftex]{hyperref}

\else

    \usepackage[dvips]{epsfig}
    \newcommand{\href}[2]{#2}

\fi

\theoremstyle{definition}
\newtheorem{theorem}{Theorem}[section]

\newtheorem{definition}[theorem]{Definition}

\def\compactify{\itemsep=0pt \topsep=0pt \partopsep=0pt \parsep=0pt}
\let\latexusecounter=\usecounter
\newenvironment{itemize*}
  {\def\usecounter{\compactify\latexusecounter}
   \begin{itemize}}
  {\end{itemize}\let\usecounter=\latexusecounter}
\newenvironment{enumerate*}
  {\def\usecounter{\compactify\latexusecounter}
   \begin{enumerate}}
  {\end{enumerate}\let\usecounter=\latexusecounter}
\newenvironment{description*}
  {\begin{description}\compactify}
  {\end{description}}

\begin{document}

\title{Self-Assembly of Shapes at Constant Scale using Repulsive Forces\footnotetext{This research was supported in part by the National Science Foundation Grant CCF-1555626.}}
\author{
Austin Luchsinger\footnotemark[1]
\and
Robert Schweller\footnotemark[1]
\and
Tim Wylie\footnotemark[1]
}
\date{}
\clearpage\maketitle
\thispagestyle{empty}

\vspace*{-.5cm}
\begin{center}
$^*$Department of Computer Science\\ University of Texas - Rio Grande Valley \\
{\small \{austin.luchsinger01,robert.schweller,timothy.wylie\}@utrgv.edu} \\

\end{center}

\begin{abstract}
The algorithmic self-assembly of shapes has been considered in several
models of self-assembly. For the problem of \emph{shape construction}, we consider an extended
version of the Two-Handed Tile Assembly Model (2HAM), which contains positive (attractive)
and negative (repulsive) interactions. As a result, portions of an
assembly can become unstable and detach. In this model, we utilize
fuel-efficient computation 
to perform Turing machine
simulations for the construction of the shape.
In this paper, we show how an arbitrary shape can be constructed using an
asymptotically optimal number of distinct tile types (based on the
shape's Kolmogorov complexity). We achieve this at $O(1)$
scale factor in this straightforward model, whereas all previous results 
with sublinear scale factors utilize powerful self-assembly models containing features such as staging, tile deletion, chemical reaction networks, and tile activation/deactivation.
Furthermore, the computation and construction in our result only creates constant-size garbage assemblies as a
byproduct of assembling the shape.

\end{abstract}


\newpage
\setcounter{page}{1}

\section{Introduction} \label{sec:introduction}

A fundamental question within the field of self-assembly, and perhaps the most fundamental, is how to efficiently self-assemble general shapes with the smallest possible set of system monomers.  This question has been considered in multiple models of self-assembly.  Soloveichek and Winfree~\cite{SolWin07} first showed that any shape $S$, if scaled up sufficiently, is self-assembled within the \emph{abstract tile assembly model} (aTAM) using $O(\frac{K(S)}{\log K(S)})$ tile types, where $K(S)$ denotes the \emph{Kolmogorov} or \emph{descriptional} complexity of shape $S$ with respect to some universal Turing machine, which matches the lower bound for this problem.  This seminal result presented a concrete connection between the descriptional complexity of a shape and the efficiency of self-assembling the shape, and represents an elegant example of the potential connections between algorithmic processes and the self-assembly of matter.  The only drawback with this result is the extremely large scale factor required by construction:  the scale factor to build a shape $S$ is at least linear in $|S|$, and is typically far greater in their construction.  To lay claim as a true universal shape building scheme for potential experimental application, a much smaller scale factor is needed.  Unfortunately, example shapes exist (long thin rectangles for example) which prove that the aTAM cannot build all shapes at $o(|S|)$ scale in the minimum possible $O(\frac{K(S)}{\log K(S)})$ tile complexity.  This motivates the quest for small scale factors in more powerful self-assembly models.

The next result by Demaine, Patitz, Schweller, and Summers~\cite{RNAPods} considers general shape assembly within the \emph{staged RNAse} self-assembly model.  In this model, system tiles are separated into separate bins and mixed over distinct stages of the algorithm in a way that models realistic laboratory operations.  In addition, each tile type in this model is of type DNA or RNA, and the staging permits the addition of an RNAse enzyme at any step in the staging, thereby dissolving all tiles of type RNA, leaving DNA tiles untouched.  By adding the powerful operations of separate bins, sequential stages, and tile deletion, \cite{RNAPods} achieves general shape construction within optimal $O(\frac{K(s)}{\log K(S)})$ tile complexity using only a constant number of bins and stages, and only a logarithmic scale factor.  This leap in scale factor reduction constituted a great improvement, but required a very powerful model with both staging and tile dissolving.  In addition, the holy grail of $O(1)$ scale factor remained elusive.

The next entry into the quest for Kolmogorov optimal shape assembly at small scale comes from a recent work by Schiefer and Winfree~\cite{Schiefer2015}.  Schiefer and Winfree introduce the \emph{chemical reaction network tile assembly model} (CRN-TAM) in which chemical reaction networks and abstract tile assembly systems combine and interact by allowing CRN species to activate and deactivate tiles, while tile attachments may introduce CRN species.  This powerful interaction allowed the construction of Kolmogorov optimal systems for the assembly of general shapes at $O(1)$ scale.  Although the result provides a great scale factor, the CRN-TAM constitutes a substantial jump in model complexity and power.

In this paper we study the optimal shape building problem within one of the simplest, and most well studied models of self-assembly: \emph{the two handed tile assembly model} (2HAM), where system monomers are 4-sided tiles with glue types on each edge.  Assembly in the 2HAM proceeds whenever two previously assembled conglomerations of tiles, or assemblies, collide along matching glue types whose strength sums to some temperature threshold.  Our only addition to the model is the allowance of negative strength (i.e., \emph{repulsive}) glues, an admittedly powerful addition based on recent work~\cite{Doty2013,rgTAM,SS2013FEC,PRS2016RMN,REIF20111592}, but an addition motivated by biology~\cite{Rothemund01022000} that maintains the \emph{passive} nature of the model as system monomers are static, state-less pieces that simply attract or repulse based solely on surface chemistry (Figure \ref{fig:modelex}).  The negative glue 2HAM is one of the simplest models of late to consider the general shape assembly problem, and our result is on par with the best possible result: we show that any connected shape $S$ is self-assembled at $O(1)$-scale in the negative glue 2HAM within $O(\frac{K(S)}{\log K(S)})$ tile types, which is met by a matching lower bound.

\textbf{Our Approach.}
We achieve our result by combining the \emph{fuel efficient Turing machine} construction published in SODA 2013, \cite{SS2013FEC}, with a number of novel negative glue based gadgets.  At a high level, the fuel efficient Turing machine system extracts a description of a path that walks the pixels of the constant-scaled shape from a compressed initial binary string.  From there, the steps of the path are translated into \emph{walker} gadgets which conceptually walk along the surface of the growing path and eventually deposit an additional pixel in the specified direction, with the aid of \emph{path extension} gadgets. When all pixels have been placed, the path through the shape is filled, resulting in a scaled version of the original shape.

\textbf{Additional Related Work.}  Additional work has considered assembly of $O(1)$-scaled shapes by breaking the assembly process up into a number of distinct stages.  In particular, \cite{DDFIRSS07} introduce the \emph{staged self-assembly} model in which intermediate tile assemblies grow in separate bins and are mixed and split over a sequence of distinct stages.  This approach is applied to achieve $O(1)$-scaled shapes with $O(1)$ tiles types, but a large number of bins and stages which encode the target shape.  In~\cite{DFS2015NGA} this approach is pushed further to achieve tradeoffs in terms of bin complexity and stage complexity, while maintaining construction of a final assembly with no unbonded edges.  In~\cite{MSS2012SWT} similar constant-scale results are obtained in the \emph{step-wise self-assembly} model in which tile sets are added in sequence to a growing seed assembly.  Finally, in~\cite{Sum09} $O(1)$-scaled shapes are assembled with $O(1)$ tile types by simply adjusting the temperature of a given system over multiple assembly stages.  While each of above \emph{staged} approaches offers important algorithmic insights, the large number of stages required by each makes the approaches infeasible for large shapes.  Furthermore, the system complexity of these systems (which includes the staging algorithms) greatly exceeds the descriptional complexity of the goal shape in a typical case.

%
%

%
%

\textbf{Paper layout.} Our construction consists of a number of detailed gadgets for specific tasks.  Presentation is thus organized incrementally to walk through a version of each gadget (with symmetry there may be multiple). Section \ref{sec:definitions} gives the preliminary definitions and background. In Section \ref{sec:overview} we provide a high-level overview of the entire process as a guide for the rest of the paper. The details begin in Section \ref{sec:lineconstruction} with the construction gadgets and how to construct a line of the path. Section \ref{sec:turning} then covers turning corners. We show how the shape is filled in Section \ref{sec:filling} to complete the construction. Section \ref{sec:shapes} provides the analysis of our construction, with the lower bound on tile complexity for shape assembly presented in Section~\ref{sec:lowerBound}, and details for pushing our construction to achieve a matching upper bound in Section~\ref{sec:baseConversion}. Then we conclude in \ref{sec:conclusion}.

\section{Definitions and Model}\label{sec:definitions}
\begin{figure*}[t]
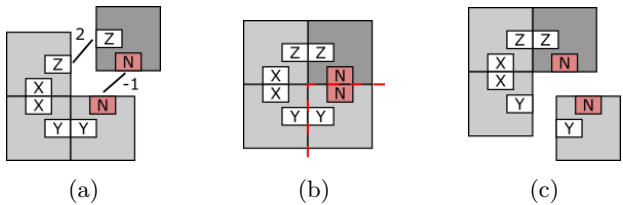

	\begin{center}
	\begin{subfigure}[b]{.15\textwidth}
		\centering
		\includegraphics[scale=6.0]{Definitions/modelex_1.pdf}
		\caption{}
	\end{subfigure}
	$\quad$
	\begin{subfigure}[b]{.15\textwidth}
		\centering
		\includegraphics[scale=6.0]{Definitions/modelex_2.pdf}
		\caption{}
	\end{subfigure}
	$\quad$
	\begin{subfigure}[b]{.15\textwidth}
		\centering
		\includegraphics[scale=6.0]{Definitions/modelex_3.pdf}
		\caption{}
	\end{subfigure}
	\caption{This figure introduces notation for our constructions, as well as a simple example of attachment and detachment events. On each tile, the glue label is presented. Red (shaded) labels represent negative glues, and the relevant glue strengths for the tiles can be found in the captions. For caption brevity, for a glue type $X$ we denote $str(X)$ simply as $X$ (e.g., $X + Y = str(X) + str(Y)$). In this temperature $\tau=1$ example, $X = 2$, $Y=1$, $Z=2$ and $N=-1$. (a) The three tile assembly on the left attaches with the single tile with strength $Z + N = 2  - 1 = \tau$ resulting in the $2\times2$ assembly shown in (b). However, this $2\times2$ assembly is unstable along the cut shown by the dotted line, since $Y + N = 1 - 1 < \tau$. Then the assembly is breakable into the assemblies shown in (c). }
	\label{fig:modelex}
	\end{center}
	\vspace*{-.5cm}
\end{figure*}

In this section we first define the two-handed tile self-assembly model with both negative and positive strength glue types.  We also formulate the problem of designing a tile assembly system that constructs a constant-scaled shape given the optimal description of that shape.

\subsection{Two Handed Self-Assembly Model (2HAM) with Negative Glues}
\textbf{Tiles and Assemblies.} A tile is an axis-aligned unit square centered at a point in $\mathbb{Z}^2$, where each edge is labeled by a \emph{glue} selected from a glue set $\Pi$.
A \emph{strength function} ${\rm str} : \Pi \rightarrow \mathbb{N}$ denotes the \emph{strength} of each glue.
Two tiles that are equal up to translation have the same \emph{type}.  A \emph{positioned shape} is any subset of $\mathbb{Z}^2$.  A \emph{positioned assembly} is a set of tiles at unique coordinates in $\mathbb{Z}^2$, and the positioned shape of a positioned assembly $A$ is the set of coordinates of those tiles.

For a given positioned assembly $\Upsilon$, define the \emph{bond graph} $G_\Upsilon$ to be the weighted grid graph in which each element of $\Upsilon$ is a vertex and the weight of an edge between tiles is the strength of the matching coincident glues or~0.\footnote{Note that only matching glues of the same type contribute a non-zero weight, whereas non-equal glues always contribute zero weight to the bond graph.  Relaxing this restriction has been considered as well~\cite{AGKS05g}.}  A positioned assembly $C$ is said to be \emph{$\tau$-stable} for positive integer $\tau$ provided the bond graph $G_C$ has min-cut at least $\tau$, and $C$ is said to be \emph{connected} if every pair of vertices in $G_C$ has a connecting path using only positive strength edges.

For a positioned assembly $A$ and integer vector $\vec{v} = (v_1, v_2)$, let $A_{\vec{v}}$ denote the positioned assembly obtained by translating each tile in $A$ by vector $\vec{v}$.  An \emph{assembly} is a translation-free version of a positioned assembly, formally defined to be a set of all translations $A_{\vec{v}}$ of a positioned assembly $A$.  An assembly is $\tau$-stable if and only if its positioned elements are $\tau$-stable.  An assembly is \emph{connected} if its positioned elements are connected.  A \emph{shape} is the set of all integer translations for some subset of $\mathbb{Z}^2$, and the shape of an assembly $A$ is defined to be the union of the positioned shapes of all positioned assemblies in $A$.  The \emph{size} of either an assembly or shape $X$, denoted as $|X|$, refers to the number of elements of any positioned element of $X$.

\textbf{Breakable Assemblies.}  An assembly is \emph{$\tau$-breakable} if it can be cut into two pieces along a cut whose strength sums to less than $\tau$.  Formally, an assembly $C$ is \emph{breakable} into assemblies $A$ and $B$ if $A$ and $B$ are connected, and the bond graph $G_{C'}$ for some assembly $C' \in C$ has a cut $(A',B')$ for $A'\in A$ and $B'\in B$ of strength less than $\tau$. We call $A$ and $B$ a pair of \emph{pieces} of the breakable assembly $C$.

\textbf{Combinable Assemblies.}
Two assemblies are \emph{$\tau$-combinable} provided they may attach along a border whose strength sums to at least $\tau$. Formally, two assemblies $A$ and $B$ are \emph{$\tau$-combinable} into an assembly $C$ provided $G_{C'}$ for any $C'\in C$ has a cut $(A',B')$ of strength at least $\tau$ for some $A'\in A$ and $B' \in B$.  We call $C$ a \emph{combination} of $A$ and $B$.

Note that $A$ and $B$ may be combinable into an assembly that is not stable (and thus breakable).  This is a key property that is leveraged throughout our constructions.  See Figure~\ref{fig:modelex} for an example.  For a system $\Gamma = (T,\tau)$, we say $A \rightarrow^{\Gamma}_1 B$ for assemblies $A$ and $B$ if either $A$ is $\tau$-breakable into pieces that include $B$, or $A$ is $\tau$-combinable with some producible assembly, or if $A=B$.  Intuitively this means that $A$ may grow into assembly $B$ through one or fewer combination or break reactions.  We define the relation $\rightarrow^{\Gamma}$ to be the transitive closure of $\rightarrow^{\Gamma}_1$, ie., $A \rightarrow^{\Gamma} B$ means that $A$ may grow into $B$ through a sequence of combination or break reactions.

\textbf{Producibility and Unique Assembly.}
A \emph{two-handed tile assembly system (2HAM system)} is an ordered pair $(T,\tau)$ where $T$ is a set of single tile assemblies, called the \emph{tile set}, and $\tau \in \mathbb{N}$ is the \emph{temperature}.
Assembly proceeds by repeated combination of assembly pairs, or breakage of unstable assemblies, to form new assemblies starting from the initial tile set.  The \emph{producible assemblies} are those constructed in this way.
Formally:

\begin{definition}[2HAM Producibility]
    For a given 2HAM system $\Gamma =(T,\tau)$, the set of \emph{producible assemblies} of $\Gamma$, denoted $\texttt{PROD}_\Gamma$, is defined recursively:
    \begin{itemize}
        \setlength\itemsep{.1em}
        \item (Base) $T \subseteq \texttt{PROD}_\Gamma$
        \item (Combinations) For any $A,B \in \texttt{PROD}_\Gamma$ such that $A$ and $B$ are $\tau$-combinable into $C$, then $C\in\texttt{PROD}_\Gamma$.
        \item (Breaks) For any assembly $C \in \texttt{PROD}_\Gamma$ that is $\tau$-breakable into $A$ and $B$, then $A,B \in \texttt{PROD}_\Gamma$.
    \end{itemize}
\end{definition}

\begin{definition}[Terminal Assemblies]
A \emph{terminal} assembly of a 2HAM system is a producible assembly that can not break and can not combine with any other producible assembly.  Formally, an assembly $A \in \texttt{PROD}_\Gamma$ of a 2HAM system $\Gamma =(T,\tau)$ is \emph{terminal} provided $A$ is $\tau$-stable (will not break) and not $\tau$-combinable with any producible assembly of $\Gamma$ (will not combine).
\end{definition}


\begin{definition}[Unique Assembly - with bounded garbage]
A 2HAM system \emph{uniquely} produces an assembly $A$ if all producible assemblies have a forward growth path towards the terminal assembly $A$, with the possible exception of some $O(1)$-sized producible assemblies.  Formally, a 2HAM system $\Gamma =(T,\tau)$ \emph{uniquely} produces an assembly $A$ provided that $A$ is terminal, and for all $B \in \texttt{PROD}_\Gamma$ such that $|B| \geq c$ for some constant $c$, $B \rightarrow^{\Gamma} A$.
\end{definition}

\begin{definition}[Unique Shape Assembly - with bounded garbage]
A 2HAM system uniquely produces a shape $S$ if all producible assemblies have a forward growth path to a terminal assembly of shape $S$ with the possible exception of some $O(1)$-sized producible assemblies.  Formally, a 2HAM system $\Gamma =(T,\tau)$ \emph{uniquely assembles} a finite shape $S$ if for every $A \in \texttt{PROD}_\Gamma$ such that $|A|\leq c$ for some constant $c$, there exists a terminal $A' \in \texttt{PROD}_\Gamma$ of shape $S$ such that $A \rightarrow^{\Gamma} A'$.
\end{definition}

\subsection{Descriptional Complexity}
\begin{definition}[Kolmogorov Complexity]
The \emph{Kolmogorov complexity} (or \emph{descriptional complexity}) of a shape $S$ with respect to some fixed universal Turing machine $U$ is the smallest bit string such that $U$ outputs a list of exactly the positions in some translation of shape $S$ when provided the bit string as input.  We denote this value as $K(S)$.
\end{definition} 

\section{Concept/Construction Overview} \label{sec:overview}
This section presents a high-level overview of the shape construction process. First, we will present the conceptual overview, which explains the fundamental ideas behind our shape self-assembly process. Then, we will show a high-level look at how our construction implements this process.
\subsection{Conceptual Overview}

Starting with the Kolmogorov-optimal description of a shape (as a base $b$ string, $b>2$), we simulate a Turing machine which converts any base $b$ string into its equivalent base $2$ representation (Sec.~\ref{sec:baseConversion})
We then simulate another Turing machine that takes the binary description of a shape, finds a spanning tree for that shape, and outputs a path around that spanning tree as a set of instructions (forward, left, right) starting from a beginning node on the perimeter.

A simple depth-first search will find the spanning tree for any shape. Scaling the shape to scale 2 creates a perimeter \emph{path} that outlines the spanning tree, and assembles the shape. Scaling again, this time by a multiple of 3, now allows space for the perimeter path with an equal-sized space buffer on both sides (Fig.~\ref{fig:scaling}).

\begin{figure}[t]
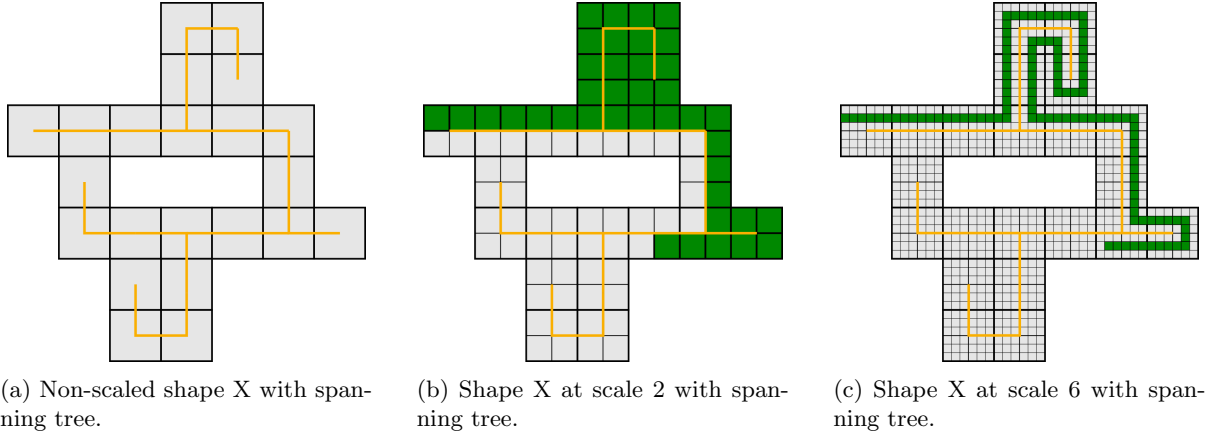

	\centering
	\begin{subfigure}[b]{.3\textwidth}
		\centering
		\includegraphics[scale=2.0]{Overview/scaling_1.pdf}
		\caption{Non-scaled shape X with spanning tree.}
	\end{subfigure}
	$\quad$
	\begin{subfigure}[b]{.3\textwidth}
		\centering
		\includegraphics[scale=2.0]{Overview/scaling_2.pdf}
		\caption{Shape X at scale 2 with spanning tree.}
	\end{subfigure}
	$\quad$
	\begin{subfigure}[b]{.3\textwidth}
		\centering
		\includegraphics[scale=2.0]{Overview/scaling_3.pdf}
		\caption{Shape X at scale 6 with spanning tree.}
	\end{subfigure}
	\caption{The Turing machine calculates a spanning tree of the tiles in the shape (a), scales the shape in order to allow a path around the spanning tree (b), and further scales the shape for the gadgets (c). }
	\label{fig:scaling}
\end{figure}

\paragraph{Process Overview:}
\vspace*{-.2cm}
\begin{enumerate}
    \setlength\itemsep{.1em}
  \item Given the Kolmogorov-optimal description of a shape, run a base conversion Turing machine to get its binary equivalent.
  \item Given that binary string, run another Turing machine that outputs the description of a path around the shape's spanning tree as a set of instructions (forward, left, right).
  \item Given those instructions, build the path. Our construction begins with a \emph{tape} containing this \emph{path} description for a scale 24 shape.
\end{enumerate}

\subsection{Construction Overview} \label{subsec:conover}
The construction overview begins at step 3 of the conceptual overview, using the output from step 2. Throughout this paper, we will be referring to this output as the \emph{tape}, meaning the fuel-efficient Turing machine tape with \emph{path}-building instructions encoded on it. This \emph{tape} is detailed in Section~\ref{sec:lineconstruction}.
\begin{figure}[t]
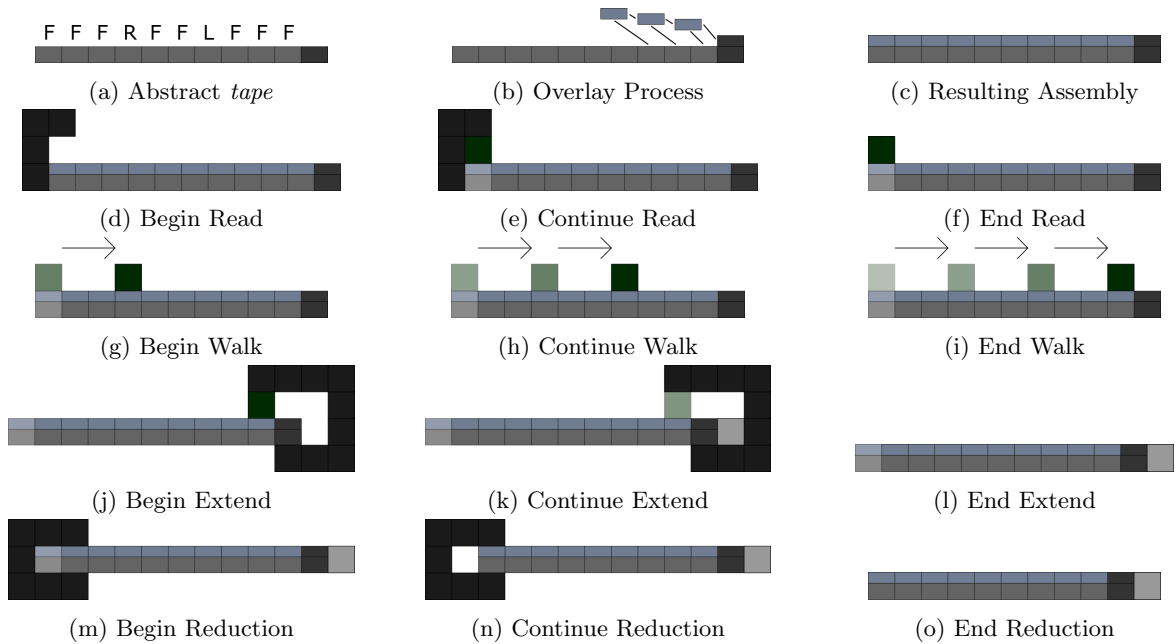

	\centering
	\begin{subfigure}[b]{.3\textwidth}
		\centering
		\includegraphics[scale=2.5]{Overview/overlay_1.pdf}
		\caption{Abstract \textit{tape}}
	\end{subfigure}
	$\quad$
	\begin{subfigure}[b]{.3\textwidth}
		\centering
		\includegraphics[scale=2.5]{Overview/overlay_2.pdf}
		\caption{Overlay Process}
	\end{subfigure}
	$\quad$
	\begin{subfigure}[b]{.3\textwidth}
		\centering
		\includegraphics[scale=2.5]{Overview/overlay_3.pdf}
		\caption{Resulting Assembly}
	\end{subfigure}\\
	\begin{subfigure}[b]{.3\textwidth}
		\centering
		\includegraphics[scale=2.5]{Overview/read_1.pdf}
		\caption{Begin Read}
	\end{subfigure}
	$\quad$
	\begin{subfigure}[b]{.3\textwidth}
		\centering
		\includegraphics[scale=2.5]{Overview/read_2.pdf}
		\caption{Continue Read}
	\end{subfigure}
	$\quad$
	\begin{subfigure}[b]{.3\textwidth}
		\centering
		\includegraphics[scale=2.5]{Overview/read_3.pdf}
		\caption{End Read}
	\end{subfigure}\\
	\begin{subfigure}[b]{.3\textwidth}
		\centering
		\includegraphics[scale=2.5]{Overview/walk_1.pdf}
		\caption{Begin Walk}
	\end{subfigure}
	$\quad$
	\begin{subfigure}[b]{.3\textwidth}
		\centering
		\includegraphics[scale=2.5]{Overview/walk_2.pdf}
		\caption{Continue Walk}		
	\end{subfigure}
	$\quad$
	\begin{subfigure}[b]{.3\textwidth}
		\centering
		\includegraphics[scale=2.5]{Overview/walk_3.pdf}
		\caption{End Walk}		
	\end{subfigure}\\
	\begin{subfigure}[b]{.3\textwidth}
		\centering
		\includegraphics[scale=2.5]{Overview/extend_1.pdf}
		\caption{Begin Extend}
	\end{subfigure}
	$\quad$
	\begin{subfigure}[b]{.3\textwidth}
		\centering
		\includegraphics[scale=2.5]{Overview/extend_2.pdf}
		\caption{Continue Extend}		
	\end{subfigure}
	$\quad$
	\begin{subfigure}[b]{.3\textwidth}
		\centering
		\includegraphics[scale=2.5]{Overview/extend_3.pdf}
		\caption{End Extend}		
	\end{subfigure}\\
		\begin{subfigure}[b]{.3\textwidth}
		\centering
		\includegraphics[scale=2.5]{Overview/reduce_1.pdf}
		\caption{Begin Reduction}
	\end{subfigure}
	$\quad$
	\begin{subfigure}[b]{.3\textwidth}
		\centering
		\includegraphics[scale=2.5]{Overview/reduce_2.pdf}
		\caption{Continue Reduction}		
	\end{subfigure}
	$\quad$
	\begin{subfigure}[b]{.3\textwidth}
		\centering
		\includegraphics[scale=2.5]{Overview/reduce_3.pdf}
		\caption{End Reduction}		
	\end{subfigure}\\
	\caption{(a)-(c) The \emph{overlay} process covers the tape while making the data readable on top.
	(d)-(f) \emph{Reading} the leftmost piece of data and creating an information block.
	(g)-(i) \emph{Information Walking} on the path to the end where the information is used.
	(j)-(l) When the information block reaches the end of the path, the block triggers a \emph{Path Extension}.
	(m)-(o) Once the information has been read, \emph{Tape Reduction} removes that piece of the tape.}
	\label{fig:combover1}
\end{figure}

\paragraph{Construction Steps Overview:}
\vspace*{-.2cm}
\begin{enumerate}
    \item \textbf{Overlay.}
The overlay process is the first step in shape construction. Figure~\ref{fig:combover1}a-c 
shows an abstraction of how the output from step 2 in the concept overview gets covered during overlay process. The overlay initiator gadget can only attach to a completed tape. This begins a series of cooperative attachments that will cover the tape. Each bit of information on the tape is covered by its corresponding overlay piece, so the information is readable on the top of the overlay. The overlay process is finished once the entire tape is covered.

\item \textbf{Reading.}
After the overlay process is complete, information can be extracted from the tape through the read process (Figs. \ref{fig:combover1}d-f). 
Information can only be extracted from the covered leftmost section of the tape if it has not already been read. When a tape section is read, information is extracted from the tape and a corresponding information block is created.

\item \textbf{Information Walking.}
Once the information block is created, it begins walking until it reaches the end of the tape/path (Figs. \ref{fig:combover1}g-i). 
Walking gadgets allow the information to travel down the entire path.

\item \textbf{Path Extension.}
When an information block cannot travel any further, the path is extended (Figs. \ref{fig:combover1}j-l). 
The path can be extended forward, left, or right. The direction of the path extension is dependent on which information block is at the end of the path. After the path is extended, the information block is removed from the path.

\item \textbf{Tape Reduction.}
Once information is extracted from the tape and sent down the path, one tape section is removed (Figs. \ref{fig:combover1}j-l). 
Only tape sections that have been read are removed, which then allows the next section to be read. This process continues until every section of the tape is read/removed.

\begin{figure}[t!]
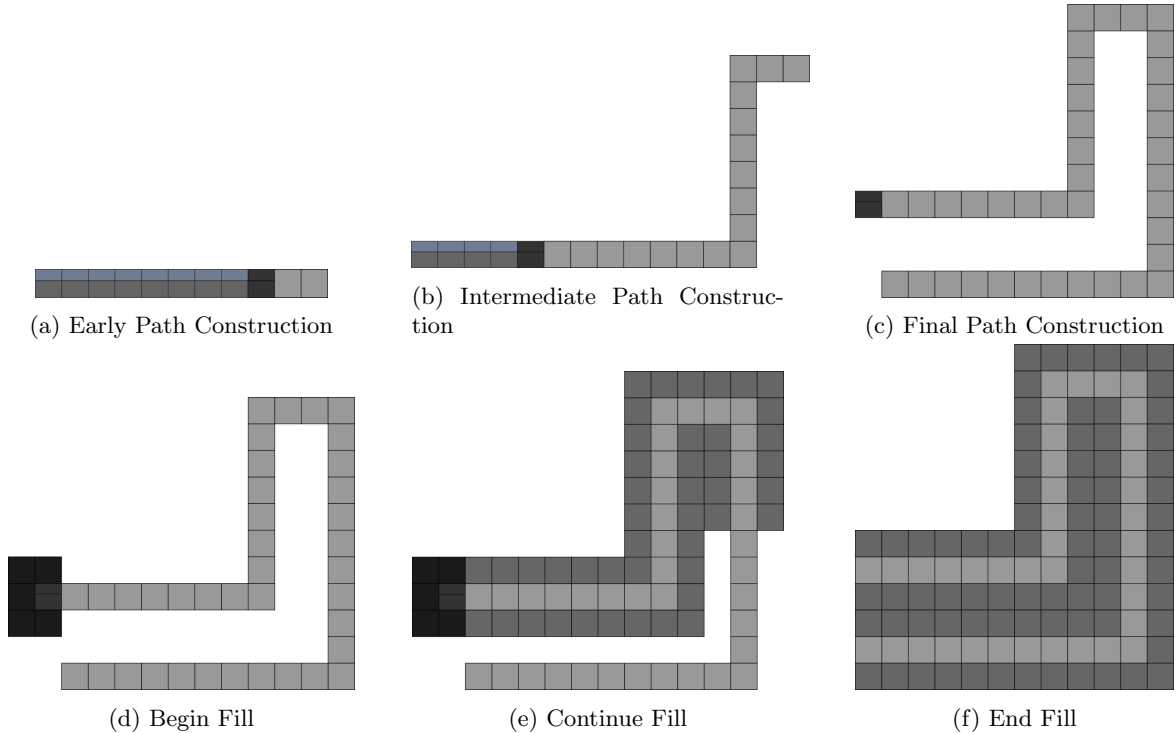

	\centering
	\begin{subfigure}[b]{.3\textwidth}
		\centering
		\includegraphics[scale=2.5]{Overview/continue_1.pdf}
		\caption{Early Path Construction}
	\end{subfigure}
	$\quad$
	\begin{subfigure}[b]{.3\textwidth}
		\centering
		\includegraphics[scale=2.5]{Overview/continue_2.pdf}
		\caption{Intermediate Path Construction}		
	\end{subfigure}
	$\quad$
	\begin{subfigure}[b]{.3\textwidth}
		\centering
		\includegraphics[scale=2.5]{Overview/continue_3.pdf}
		\caption{Final Path Construction}
	\end{subfigure}\\
\begin{subfigure}[b]{.3\textwidth}
		\centering
		\includegraphics[scale=2.5]{Overview/fill_1.pdf}
		\caption{Begin Fill}
	\end{subfigure}
	$\quad$
	\begin{subfigure}[b]{.3\textwidth}
		\centering
		\includegraphics[scale=2.5]{Overview/fill_2.pdf}
		\caption{Continue Fill}		
	\end{subfigure}
	$\quad$
	\begin{subfigure}[b]{.3\textwidth}
		\centering
		\includegraphics[scale=2.5]{Overview/fill_3.pdf}
		\caption{End Fill}		
	\end{subfigure}\\	
	\caption{(a)-(c) The process is repeated until all information has been read/removed from the tape.
	(d)-(f) The final step is \emph{Path Filling} the shape.}
	\label{fig:combover2}
\end{figure}

\item \textbf{Repeat.}
Repeat the tape read, information walk, path extend, and tape reduction processes until all path instructions have been read (Figs. \ref{fig:combover2}a-c). 

\item \textbf{Path Filling.}
The final tape section that gets read begins the shape fill process (Figs. \ref{fig:combover2}d-f). 
In this process, the path is padded with tiles which fill it in and results in the final shape.

\end{enumerate}


\section{Forward Path Construction} \label{sec:lineconstruction}
In this section, we detail the steps presented in the construction overview (Sec.~\ref{subsec:conover}). This is the process by which information is read from the \emph{tape} and portions of the \emph{path} are assembled. For clarity, in this section we will just be showing how \emph{forward} instructions are read in order to build a linear section of the \emph{path}. Although each glue strength can be found in the figure captions, there is a full table of glue strengths in Table~\ref{tbl:gluetable}.

\begin{table}[H]
\begin{center}
\begin{tabular}{ | c | c | c | c | }
    \hline
    \textbf{Label}           &   \textbf{Strength}  &   \textbf{Label}       &   \textbf{Strength} \\ \hline
F,L,R,M             &   1         &   O,T         &   7        \\ \hline
A,X,f,k,l,r,p,w,h   &   2         &   G,H,J,U,c,g,j,m,s,t,u,v,x,z  &   8        \\ \hline
B,b,e,$f^*$         &   3         &   K,P,V,Y,Z   &   9       \\ \hline
C,a                 &   4         &   Q           &   -4       \\ \hline
N,E,S,W,i,q         &   5         &   o           &   -5         \\ \hline
d                   &   6         &   D           &   -7       \\ \hline
\end{tabular}
\caption{The glue strengths of each glue label in the construction system.}
\label{tbl:gluetable}
\end{center}
\end{table} 

We will also cover the gadgets required for each step, and review the tape construction from the fuel-efficient Turing machine used in \cite{SS2013FEC}. This construction uses pre-constructed assemblies called gadgets. These gadgets are designed to work in a temperature $\tau = 10$ system. In our figures, a perpendicular black line through the middle of the edge of two adjacent tiles indicates a unique infinite strength bond (i.e. the strength of the glue is $\gg$ than $\tau$ such that no detachment events can occur in which these tiles are separated). Each gadget provides a different function to the shape creation process. Some gadgets have special-case versions which are very similar to their standard counterparts. For clarity of understanding, in this section we will only be describing the standard version of each gadget. For a description of all special case gadgets, see Appendix~\ref{app:specialgadgets}.


\begin{figure}[t]
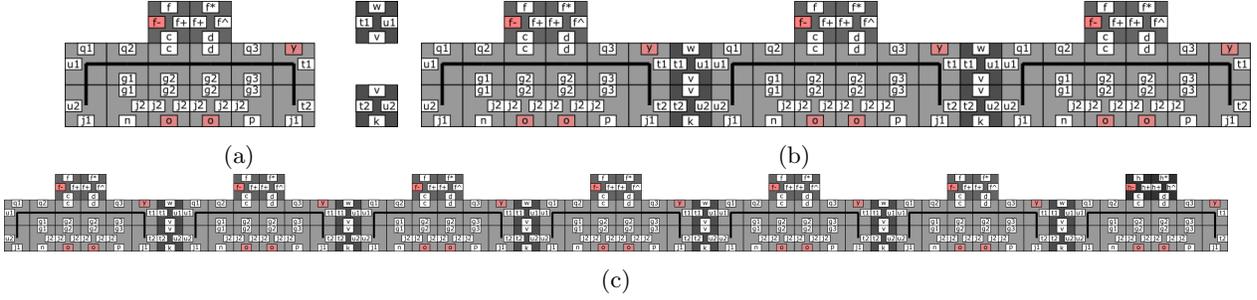

	\centering
	\begin{subfigure}[b]{.28\textwidth}
		\includegraphics[scale=3.9]{ConstructionGadgets/tape_1.pdf}
		\caption{}
	\end{subfigure}
	\begin{subfigure}[b]{.6\textwidth}
		\includegraphics[scale=3.9]{ConstructionGadgets/tape_2.pdf}
		\caption{}
	\end{subfigure}\\
	\begin{subfigure}[b]{1\textwidth}
		\centering
		\includegraphics[scale=2.4]{ConstructionGadgets/tape_3.pdf}
		\caption{}
	\end{subfigure}
	\caption{(a) A single section of a fuel-efficient Turing machine \emph{tape} used in this construction, along with buffer tiles that separate tape sections. (b) A series of joined tape sections. (c) A completed \emph{tape} consisting of all \emph{forward} instructions.}
	\label{fig:tapesections}
\end{figure}

\paragraph{Turing Machine Tape.}
A detailed look at a fuel-efficient Turing machine \emph{tape} is seen in Figure~\ref{fig:tapesections}. Notice each tape section has a pair of tiles on top of it where the information is stored. When talking about the \emph{tape} from Section~\ref{subsec:conover}, each pair of dark grey tiles on top of the tape sections represents a piece of information describing the \emph{path}.


\begin{figure}[]
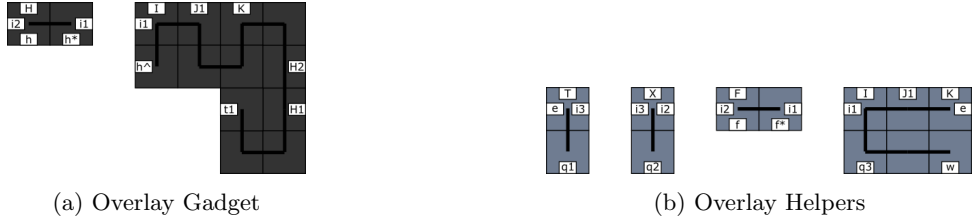

	\centering
	\begin{subfigure}[b]{.45\textwidth}
		\centering
		\includegraphics[scale=4.0]{ConstructionGadgets/overlay_gadget.pdf}
		\caption{Overlay Gadget}
	\end{subfigure}
	$\quad$
	\begin{subfigure}[b]{.45\textwidth}
		\centering
		\includegraphics[scale=4.0]{ConstructionGadgets/overlay_helpers.pdf}
		\caption{Overlay Helpers}
	\end{subfigure}
	\caption{(a) The overlay gadget comes in two parts, and attaches to the tape using its \emph{h}-type glues, which are unique to the end of a completed Turing machine tape. (b) Four filler blocks make up the overlay helpers. Three of these helpers are standard fillers, but the horizontal domino piece has multiple versions to cover the different information bits on the tape.}
	\label{fig:overlaygadget}
\end{figure}

\begin{figure}[t]
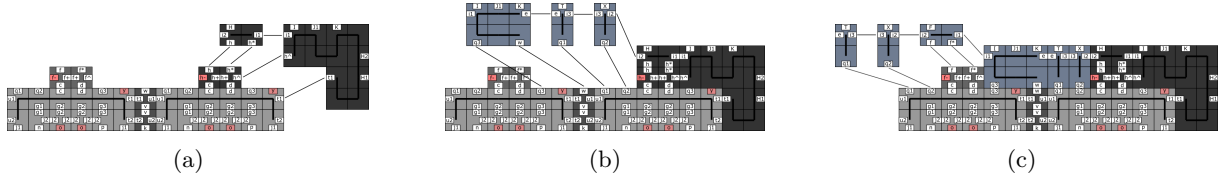

	\centering
	\begin{subfigure}[b]{.3\textwidth}
		\centering
		\includegraphics[scale=2.0]{LineConstruction/line_overlay_1.pdf}
		\caption{}
	\end{subfigure}
	$\quad$
	\begin{subfigure}[b]{.3\textwidth}
		\centering
		\includegraphics[scale=2.0]{LineConstruction/line_overlay_2.pdf}
		\caption{}
	\end{subfigure}
	$\quad$
	\begin{subfigure}[b]{.3\textwidth}
		\centering
		\includegraphics[scale=2.0]{LineConstruction/line_overlay_3.pdf}
		\caption{}
	\end{subfigure}
	\caption{(a) The Overlay Initiator Gadget attaches to the end of the tape ($t1 + $ \textit{h\^} $ = 8 + 2 \geq \tau = 10$), and then an overlay tile, unique to the end of the tape, attaches ($h + h\ast + i1 = 2 + 3 + 5 \geq \tau$). This begins the rest of the overlay process. Two vertical domino assemblies attach ($q2 + i2 = 5 + 5 \geq \tau, q1 + i3 = 5 + 5 \geq \tau $), and finish covering the first tape section as seen in (b). A 2x3 overlay block attaches ($q3 + w + e = 5 + 2 + 3 \geq \tau$), and covers the first portion of the next tape section. At this point, one of three horizontal domino overlay assemblies attach depending on which information bit is on the tape section. In (c), since the $f$ bit is on the tape, it's corresponding overlay domino attaches ($f + f\ast + i1 = 2 + 3 + 5 \geq \tau $) which carries the \textit{forward} instruction. Two more vertical dominos finish covering the tape section, and this process is repeated until every section of the tape is covered.}
	\label{fig:lineoverlay}
\end{figure}

\paragraph{Overlay.}
The overlay gadget (Fig.~\ref{fig:overlaygadget}) is the first assembly to attach to the completed \emph{tape}, and allows the attachment of subsequent assemblies (overlay helpers) which end up covering the entire \emph{tape}. In addition to being designed to initiate the overlay process (Fig. \ref{fig:lineoverlay}), the overlay gadget also has glues on its east side which will be the starting location for the \emph{path} we are constructing. Once the gadget begins the overlay process described in Section~\ref{subsec:conover}, the tape gets covered section by section. Each bit of information on the \emph{tape} has a corresponding overlay piece that will attach to it. This step essentially ``flattens'' the surface of the \emph{tape}, making the information easily readable by the gadgets used in this construction.



\begin{figure}[]
    \vspace{-1cm}
	\centering
	\begin{subfigure}[b]{.45\textwidth}
		\centering
		\includegraphics[scale=4.0]{ConstructionGadgets/read_combined.pdf}
		\caption{Read gadget and helpers}
	\end{subfigure}
	$\quad$
	\begin{subfigure}[b]{.45\textwidth}
		\centering
		\includegraphics[scale=4.0]{ConstructionGadgets/walking_combined.pdf}
		\caption{Walk gadget and helpers}
	\end{subfigure}
	\caption{(a) The read gadget (dark grey) has information-specific versions. The \emph{F} glues in this figure indicate that this gadget reads the \emph{forward} instruction. The read helpers (blue) allow the first information block (green) to attach to the top of the overlay, essentially extracting the information from the tape. (b) Like the read gadget, the walking gadget (dark grey) also has information-specific versions. The light-green assembly is the corresponding information block that this walker will be transporting. The two other helpers (olive drab) remove the previous info block/helpers and the walking gadget.}
	\label{fig:combinedgadgets1}

    \vspace{1cm}
	\begin{subfigure}[b]{.3\textwidth}
		\centering
		\includegraphics[scale=2.0]{LineConstruction/line_read_1.pdf}
		\caption{}
	\end{subfigure}
	$\quad$
	\begin{subfigure}[b]{.3\textwidth}
		\centering
		\includegraphics[scale=2.0]{LineConstruction/line_read_2.pdf}
		\caption{}
	\end{subfigure}
	$\quad$
	\begin{subfigure}[b]{.3\textwidth}
		\centering
		\includegraphics[scale=2.0]{LineConstruction/line_read_3.pdf}
		\caption{}
	\end{subfigure}
	\begin{subfigure}[b]{.3\textwidth}
		\centering
		\includegraphics[scale=2.0]{LineConstruction/line_read_4.pdf}
		\caption{}
	\end{subfigure}
	$\quad$
	\begin{subfigure}[b]{.3\textwidth}
		\centering
		\includegraphics[scale=2.0]{LineConstruction/line_read_5.pdf}
		\caption{}
	\end{subfigure}
	$\quad$
	\begin{subfigure}[b]{.3\textwidth}
		\centering
		\includegraphics[scale=2.0]{LineConstruction/line_read_6.pdf}
		\caption{}
	\end{subfigure}
	\caption{(a) The Read Gadget attaches ($n + T + F = 2 + 7 + 1 \geq \tau$). (b) This allows the first read-helpers to attach ($K + M = 9 + 1 \geq \tau, J1 + A1 = 8 + 2 \geq \tau$). In (c) the first form of an information block attaches ($F + F + J2 = 1 + 1 + 8 \geq \tau$). Since the \textit{forward} version of the read gadget was used, the \textit{forward} information block is placed. After the information block is placed, the penultimate read-helper attaches ($A2 + A2 + O1 = 2 + 2 +7 \geq \tau$). (d) The final read-helper attaches ($J3 + J3 + Q = 8 + 8 - 4 \geq \tau$). Notice the attraction of the $J3$ glues overcomes the negativity of the $Q$ glue, allowing attachment. This, however, destabilizes the read gadget along the cut shown in (e)($F + F + M + n + T + F + Q = 1 + 1 + 1 + 2 + 7 + 1 - 7 \leq \tau$). (f) The unstable gadget falls off and the information block (by way of the read-helpers) remains attached to the tape-overlay.}
	\label{fig:lineread}
\end{figure}

\paragraph{Read.}
The read gadget (Fig.~\ref{fig:combinedgadgets1}a) is required for ``reading'' the Turing machine \emph{tape}. Essentially, this gadget extracts the information that is relayed from the \emph{tape} through the overlay blocks. The read process (Fig. \ref{fig:lineread}) can only begin if the leftmost tape section has not previously been read. Once attached, the gadget allows the attachment of an information block (corresponding to the information being read) that will be used to carry the build instructions through the rest of our construction. Once the information block is present, the remaining read-helpers can attach. The final helper destabilizes the read gadget, allowing it to fall off and expose the newly attached information block. Notice, that after the read gadget detaches, the overlay is now altered and this \emph{tape} section can no longer be read. The read gadget was designed to produce this information block, alter the \emph{tape} section that is being read (making it unreadable), and then detach from the assembly. This design ensures that each \emph{tape} section is only read once, and allows us to transfer the instructions to other locations in our construction via the walking gadgets.



\begin{figure}[]
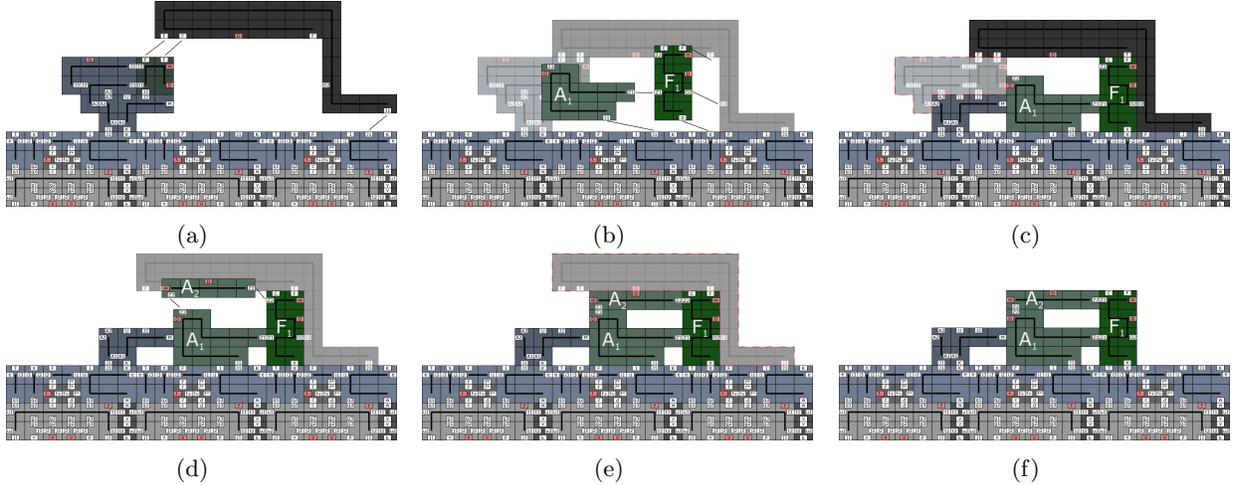

	\centering
	\begin{subfigure}[b]{.3\textwidth}
		\centering
		\includegraphics[scale=1.75]{LineConstruction/line_walk_a_1.pdf}
		\caption{}
	\end{subfigure}
	$\quad$
	\begin{subfigure}[b]{.3\textwidth}
		\centering
		\includegraphics[scale=1.75]{LineConstruction/line_walk_a_2.pdf}
		\caption{}
	\end{subfigure}
	$\quad$
	\begin{subfigure}[b]{.3\textwidth}
		\centering
		\includegraphics[scale=1.75]{LineConstruction/line_walk_a_3.pdf}
		\caption{}
	\end{subfigure}
	\begin{subfigure}[b]{.3\textwidth}
		\centering
		\includegraphics[scale=1.75]{LineConstruction/line_walk_a_4.pdf}
		\caption{}
	\end{subfigure}
	$\quad$
	\begin{subfigure}[b]{.3\textwidth}
		\centering
		\includegraphics[scale=1.75]{LineConstruction/line_walk_a_5.pdf}
		\caption{}
	\end{subfigure}
	$\quad$
	\begin{subfigure}[b]{.3\textwidth}
		\centering
		\includegraphics[scale=1.75]{LineConstruction/line_walk_a_6.pdf}
		\caption{}
	\end{subfigure}
	\caption{(a) A Walking Gadget (specific to the information block) attaches to the overlay and the information block ($F + F + +J1 = 1 + 1 + 8 \geq \tau$). (b) A new \textit{forward} information block attaches ($F + O2 + X = 1 + 7 + 2 \geq \tau$). The first walking-helper attaches to the new information block and the overlay($J1 + Z1 + D = 8 + 9 - 7 \geq \tau$). (c) The negative interaction between the $D$ glues destabilizes the old information block, along with the two walking-helpers($J2 + A2 + A2 + F + F + D = 8 + 2 + 2 + 1 + 1 - 7 \leq \tau$). (d) The old information block and walking-helpers break off, and the second walking-helper now ($Z2 + Z2 + D = 9 + 9 - 7 \geq \tau$). (e) Once the second walking-helper is attached, the walking gadget becomes unstable ($F + O2 + J1 + D = 1 + 7 + 8 - 7 \leq \tau$). (f) The walking gadget detaches and leaves the new information block with two walking-helpers attached.}
	\label{fig:linewalk1}
\end{figure}

\paragraph{Information Walking.}
The walking gadgets (Fig.~\ref{fig:combinedgadgets1}b) begin the information walking process (Fig. \ref{fig:linewalk1}), which allows instructions to travel throughout our construction. After a tape section has been read and an information block has been placed, a walking gadget can attach. Once attached, the walking gadget allows a new information block (of the same type) to attach, while also detaching the the previous information block. Notice that this detachment will always be $O(1)$ size. After the previous information is removed, the walking gadget detaches as well, allowing the new info block to interact with other gadgets. Thus, the same information has traveled from the \emph{tape}, through the overlay, and is now traveling along the \emph{tape}. This process is repeated until the information has traveled to the end of the \emph{path}, at which point it is used to construct the next \emph{path} portion. This method is desirable because it does not allow duplicate readable instructions to be attached to the path at any time.




\begin{figure}[]
	\centering
	\begin{subfigure}[b]{.45\textwidth}
		\centering
		\includegraphics[scale=4.0]{ConstructionGadgets/extension_combined.pdf}
		\caption{Extension gadget and helpers}
	\end{subfigure}
	$\quad$
	\begin{subfigure}[b]{.45\textwidth}
		\centering
		\includegraphics[scale=4.0]{ConstructionGadgets/reduction_combined.pdf}
		\caption{Reduction gadget and helpers}
	\end{subfigure}
	\caption{(a) The path extension gadget (dark grey) also has information-specific versions. The light grey assemblies attach piece-by-piece to form a path portion. The purple helpers are used to detach the extension gadget, along with the info block that allowed this extension. (b) The tape reduction gadget (dark grey) attaches once a given \emph{tape} section has already been read. The reduction helpers (red and dark grey) attach to the gadget, and build up the repulsive force that is responsible for the destabilization of the gadget, along with the previously read tape section.}
	\label{fig:combinedgadgets2}
    \vspace{1cm}
	\begin{subfigure}[b]{.3\textwidth}
		\centering
		\includegraphics[scale=1.75]{LineConstruction/line_extend_a_1.pdf}
		\caption{}
	\end{subfigure}
	$\quad$
	\begin{subfigure}[b]{.3\textwidth}
		\centering
		\includegraphics[scale=1.75]{LineConstruction/line_extend_a_2.pdf}
		\caption{}
	\end{subfigure}
	$\quad$
	\begin{subfigure}[b]{.3\textwidth}
		\centering
		\includegraphics[scale=1.75]{LineConstruction/line_extend_a_3.pdf}
		\caption{}
	\end{subfigure}
	\begin{subfigure}[b]{.3\textwidth}
		\centering
		\includegraphics[scale=1.75]{LineConstruction/line_extend_a_4.pdf}
		\caption{}
	\end{subfigure}
	$\quad$
	\begin{subfigure}[b]{.3\textwidth}
		\centering
		\includegraphics[scale=1.75]{LineConstruction/line_extend_a_5.pdf}
		\caption{}
	\end{subfigure}
	$\quad$
	\begin{subfigure}[b]{.3\textwidth}
		\centering
		\includegraphics[scale=1.75]{LineConstruction/line_extend_a_6.pdf}
		\caption{}
	\end{subfigure}
	\caption{(a) The forward-extension gadget attaches to the information block and Turing tape ($B + C + F + p = 3 + 4 + 1 + 2 \geq \tau$). (b) Once the gadget it in place, two path blocks attach ($X + H1 = 2 + 8 \geq \tau, P2 + H2 = 2 + 8 \geq \tau$). (c) The extension-helpers attach. The first ($O3 + V1 = 7 + 9 \geq \tau$), and the second ($V0 + V0 = 9 + 9 \geq \tau$).(d) The second extension-helper comes with the negative $D$ glue that causes targeted destabilization ($X + p + J1 + X + D = 2 + 2 + 8 + 2 - 7 \leq \tau$). The extension gadget and its helpers, along with the information block and its helpers are no longer stable along their tape-overlay edges. (e) Once the unstable assemblies detach, a second half of the path piece attaches ($P2 + P2 = 9 + 9 \geq \tau$). (f) The final result is a one path-pixel extension of the path.}
	\label{fig:lineextend}
\end{figure}

\paragraph{Path Extension.}
After the information block has reached the end of the \emph{path}, a path extension gadget (Figure~\ref{fig:combinedgadgets2}a) can attach to the assembly. Once attached, the gadget allows the \emph{path} extension process (Fig. \ref{fig:lineextend}) to begin, which extends the \emph{path} in a given direction (forward, left, or right) based on the instruction carried by the information block. The extension gadget ``reads'' the information block, and then extends the path in the given direction. Afterwards, the extension helpers destabilize the information block and extension gadget, causing a $O(1)$ sized detachment. We designed the extension gadget to essentially replace an instruction block with a corresponding \emph{path} portion. This design allows us to attach a $O(1)$ sized \emph{path} portion for each instruction read from the \emph{tape}.




\begin{figure}[]
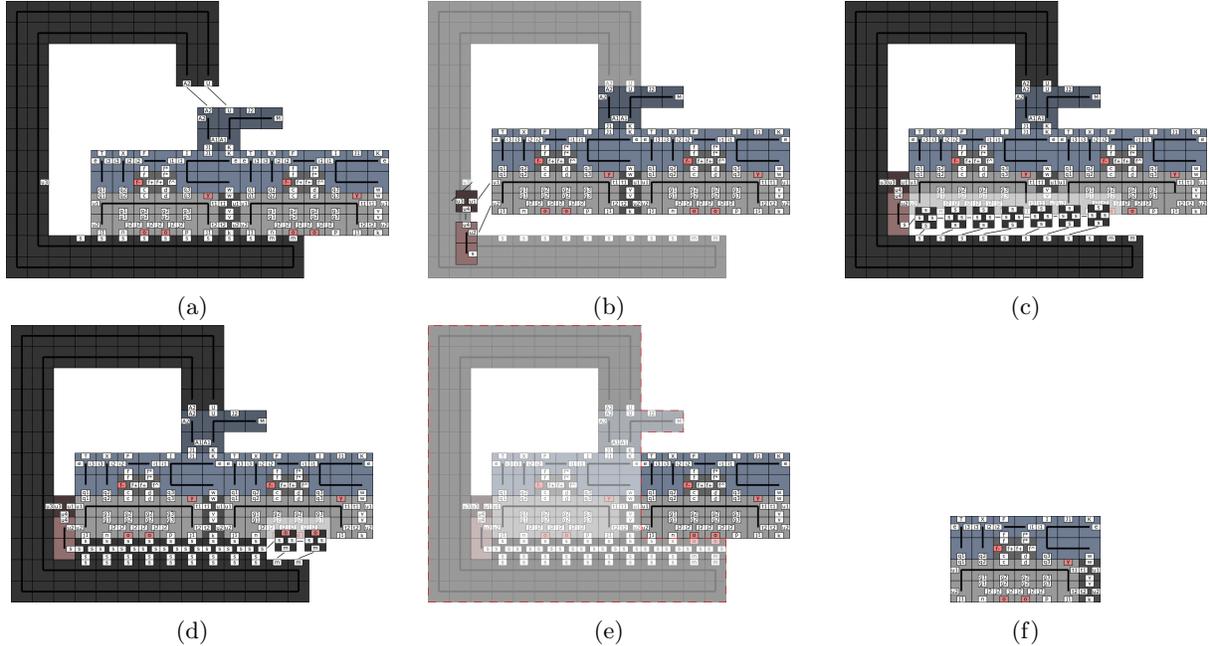

	\centering
	\begin{subfigure}[b]{.3\textwidth}
		\centering
		\includegraphics[scale=2.0]{LineConstruction/line_reduction_1.pdf}
		\caption{}
	\end{subfigure}
	$\quad$
	\begin{subfigure}[b]{.3\textwidth}
		\centering
		\includegraphics[scale=2.0]{LineConstruction/line_reduction_2.pdf}
		\caption{}
	\end{subfigure}
	$\quad$
	\begin{subfigure}[b]{.3\textwidth}
		\centering
		\includegraphics[scale=2.0]{LineConstruction/line_reduction_3.pdf}
		\caption{}
	\end{subfigure}
	\begin{subfigure}[b]{.3\textwidth}
		\centering
		\includegraphics[scale=2.0]{LineConstruction/line_reduction_4.pdf}
		\caption{}
	\end{subfigure}
	$\quad$
	\begin{subfigure}[b]{.3\textwidth}
		\centering
		\includegraphics[scale=2.0]{LineConstruction/line_reduction_5.pdf}
		\caption{}
	\end{subfigure}
	$\quad$
	\begin{subfigure}[b]{.3\textwidth}
		\centering
		\includegraphics[scale=2.0]{LineConstruction/line_reduction_6.pdf}
		\caption{}
	\end{subfigure}
	\caption{(a) The tape reduction gadget attaches to the read-helpers ($A2 + U = 2 + 8 \geq \tau$). (b) The first reduction helpers attach ($u1 + u3 = 8 + 8 \geq \tau, u2 + u4 = 8 + 8 \geq \tau$). (c) Filler tiles attach ($s + s = 8 + 8 \geq \tau$), and create a strong bond to the tape reduction gadget. (d) Two final reduction helpers attach ($s + m + o = 8 + 8 - 5 \geq \tau$). (e) The two negative \textit{o} glues cause a strong targeted destabilization of the previously read tape section ($e + u1 + u2 + o + o = 3 + 8 + 8 - 5 - 5 \leq \tau$). (f) The tape reduction gadget detaches with the used tape section, allowing the next section of the tape to be read.}
	\label{fig:linereduce}
\end{figure}

\paragraph{Tape Reduction.}
After a tape section has been read, we no longer need it. Instead of continuing to grow the assembly, we can remove $O(1)$ size portions of the \emph{tape} as it is being read. This is where the tape reduction gadget (Fig.~\ref{fig:combinedgadgets2}b) initiates the tape reduction process (Fig. \ref{fig:linereduce}) mentioned in Section~\ref{subsec:conover}. The attachments left behind by the read/walk processes allow the tape reduction gadget to attach to a tape section that has already been read. The gadget then removes itself, along with the previously read tape section, exposing the next section of the tape for reading.  This technique is desirable because it allows us to break apart the \emph{tape} into $O(1)$ sized pieces as we use it. As the \emph{tape} is reduced, the \emph{path} continues to grow until there are no more \emph{tape} sections to be read.


\section{Turning Corners} \label{sec:turning}
This section shows the details that allow the \emph{path} to turn left or right during its construction. The process by which the information is extracted and moves along the \emph{path} is identical to that of Section~\ref{sec:lineconstruction}. The key difference is how the information is used once it gets to the end of the \emph{path}. In Section~\ref{sec:lineconstruction}, forward extension and walking gadgets were used to extend the \emph{path}. Here, specific gadgets are used  in order to turn the path left and right. These gadgets are mechanically identical variations of the walking/extension gadgets shown in Section~\ref{sec:lineconstruction}.

\begin{figure}[]
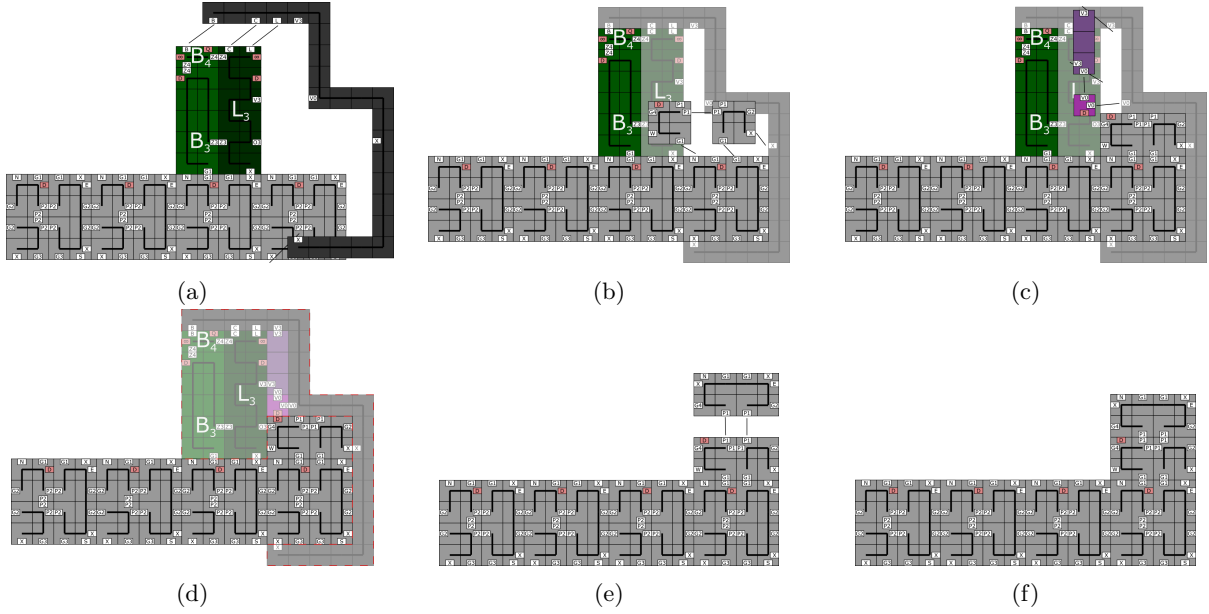

	\centering
	\begin{subfigure}[b]{.3\textwidth}
		\centering
		\includegraphics[scale=2.0]{Turning/extend_left_a1.pdf}
		\caption{}
	\end{subfigure}
	$\quad$
	\begin{subfigure}[b]{.3\textwidth}
		\centering
		\includegraphics[scale=2.0]{Turning/extend_left_a2.pdf}
		\caption{}
	\end{subfigure}
	$\quad$
	\begin{subfigure}[b]{.3\textwidth}
		\centering
		\includegraphics[scale=2.0]{Turning/extend_left_a3.pdf}
		\caption{}
	\end{subfigure}
	\begin{subfigure}[b]{.3\textwidth}
		\centering
		\includegraphics[scale=2.0]{Turning/extend_left_a4.pdf}
		\caption{}
	\end{subfigure}
	$\quad$
	\begin{subfigure}[b]{.3\textwidth}
		\centering
		\includegraphics[scale=2.0]{Turning/extend_left_a5.pdf}
		\caption{}
	\end{subfigure}
	$\quad$
	\begin{subfigure}[b]{.3\textwidth}
		\centering
		\includegraphics[scale=2.0]{Turning/extend_left_a6.pdf}
		\caption{}
	\end{subfigure}
	\caption{(a) The left-extension gadget attaches to the information block and shape path ($B + C + L + X = 3 + 4 + 1 + 2 \geq \tau$). (b) Once the gadget it in place, two path blocks attach ($X + G1 = 2 + 8 \geq \tau, P1 + G1 = 9 + 8 \geq \tau$). (c) The extension-helpers attach ($V3 + V3 = 9 + 9 \geq \tau, V0 + V0 = 9 + 9 \geq \tau $).(d) The negative \textit{D} glue on the second extension-helper causes targeted destabilization. The extension gadget and its helpers, along with the information block and its helpers are no longer bound to the path with sufficient strength. ($X + X + G1 + X + D = 2 + 2 + 8 + 2 - 7 \leq \tau$) (e) Once the unstable subassembly detaches, a second half of the path piece attaches ($P1 + P1 = 9 + 9 \geq \tau$). (f) The final result is a one path-pixel extension of the path to the left of the direction the info block was walking.}
	\label{fig:extendleft}
\end{figure}

\begin{figure}[]
	\centering
	\begin{subfigure}[b]{.3\textwidth}
		\centering
		\includegraphics[scale=2.0]{Turning/walk_left_1.pdf}
		\caption{}
	\end{subfigure}
	$\quad$
	\begin{subfigure}[b]{.3\textwidth}
		\centering
		\includegraphics[scale=2.0]{Turning/walk_left_2.pdf}
		\caption{}
	\end{subfigure}
	$\quad$
	\begin{subfigure}[b]{.3\textwidth}
		\centering
		\includegraphics[scale=2.0]{Turning/walk_left_3.pdf}
		\caption{}
	\end{subfigure}
	\begin{subfigure}[b]{.3\textwidth}
		\centering
		\includegraphics[scale=2.0]{Turning/walk_left_4.pdf}
		\caption{}
	\end{subfigure}
	$\quad$
	\begin{subfigure}[b]{.3\textwidth}
		\centering
		\includegraphics[scale=2.0]{Turning/walk_left_5.pdf}
		\caption{}
	\end{subfigure}
	$\quad$
	\begin{subfigure}[b]{.3\textwidth}
		\centering
		\includegraphics[scale=2.0]{Turning/walk_left_6.pdf}
		\caption{}
	\end{subfigure}
	\caption{The left-walking gadget attaches to the path and the information block ($C + F + W = 4 + 1 + 5 \geq \tau$). (b) A new \textit{forward} information block attaches ($F + O3 + X = 1 + 7 + 2 \geq \tau$). The first walking-helper attaches to the new information block and the path ($G1 + Z3 = 8 + 9 \geq \tau$). (c) The negative interaction between the \textit{D} glues destabilizes the old information block, along with the two walking-helpers($X + G1 + C + F + D = 2 + 8 + 4 + 1 - 7 \leq \tau$). (d) The old information block and walking-helpers break off, and the second walking-helper attaches ($Z4 + Z4 = 9 + 9 \geq \tau$). (e) Once the second walking-helper is attached, the walking gadget becomes unstable due to the negative \textit{Q} glues ($W + O3 + F + Q = 5 + 7 + 1 - 4 \leq \tau$). (f) The walking gadget detaches and leaves the new information block with two walking-helpers attached to the other side of the left turn.}
	\label{fig:walkleft}
\end{figure}

\paragraph{Extend Left.}
The process for extending left is very similar to extending forward. Just as before, the information block must walk to the end of the path before it can be used for extension. The difference now, is the direction of the extension. The \textit{left} information block allows the left-extension gadget to attach and extend the path using the same mechanics as forward-extension. (Fig. \ref{fig:extendleft}). The extension gadget reads the information block, extends the path in the given direction, and then detaches all but the newly extended path.

\paragraph{Walk Left.}
The walk-left procedure utilizes the same mechanics as walking forward, but with slightly different gadgets(Fig. \ref{fig:walkleft}). The information block only allows the correct walking gadget to attach and begin the walking process. Once attached, the walking gadget allows a new information block (of the same type) to attach, while also detaching the the previous info block. After the previous information is removed, the walking gadget detaches as well, allowing the new info block to interact with other gadgets. Thus, the same information has traveled along the left-hand path turn.

\begin{figure}[]
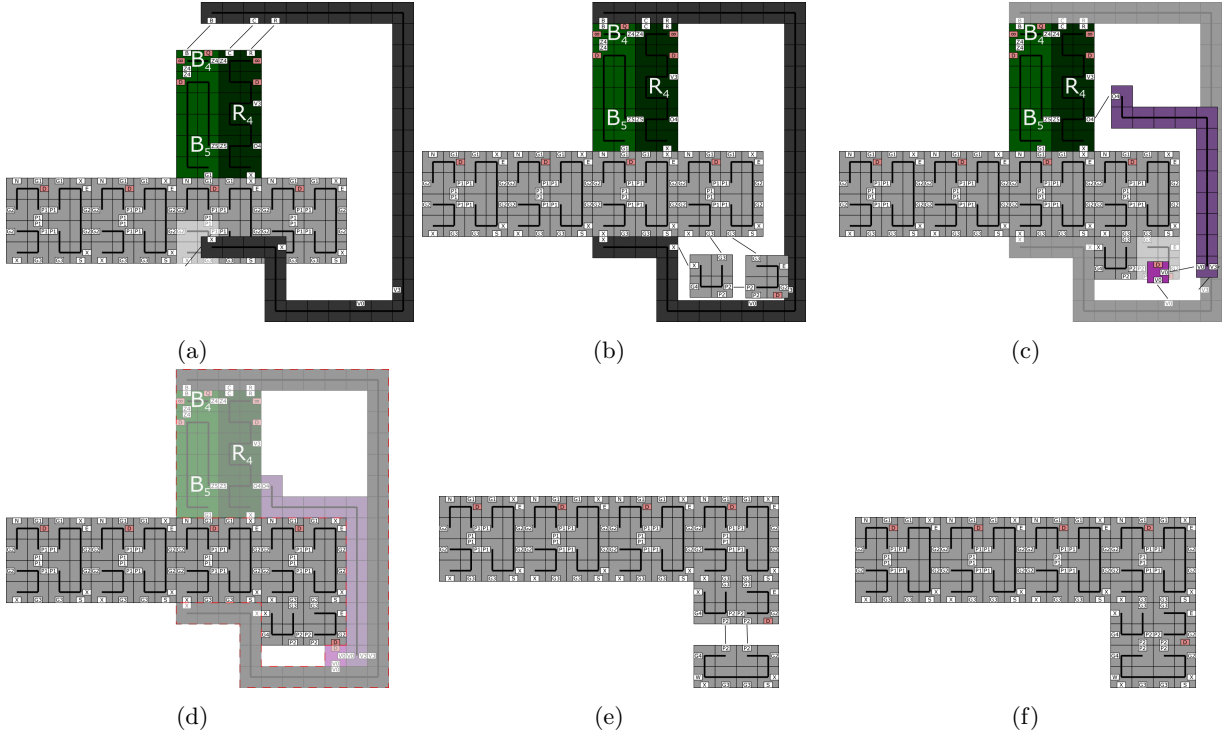

	\centering
	\begin{subfigure}[b]{.3\textwidth}
		\centering
		\includegraphics[scale=2.0]{Turning/extend_right_a1.pdf}
		\caption{}
	\end{subfigure}
	$\quad$
	\begin{subfigure}[b]{.3\textwidth}
		\centering
		\includegraphics[scale=2.0]{Turning/extend_right_a2.pdf}
		\caption{}
	\end{subfigure}
	$\quad$
	\begin{subfigure}[b]{.3\textwidth}
		\centering
		\includegraphics[scale=2.0]{Turning/extend_right_a3.pdf}
		\caption{}
	\end{subfigure}
	\begin{subfigure}[b]{.3\textwidth}
		\centering
		\includegraphics[scale=2.0]{Turning/extend_right_a4.pdf}
		\caption{}
	\end{subfigure}
	$\quad$
	\begin{subfigure}[b]{.3\textwidth}
		\centering
		\includegraphics[scale=2.0]{Turning/extend_right_a5.pdf}
		\caption{}
	\end{subfigure}
	$\quad$
	\begin{subfigure}[b]{.3\textwidth}
		\centering
		\includegraphics[scale=2.0]{Turning/extend_right_a6.pdf}
		\caption{}
	\end{subfigure}
	\caption{(a) The right-extension gadget attaches to the information block ($B + C + R + X = 3 + 4 + 1 + 2 \geq \tau$). (b) Once the gadget it in place, two path blocks attach ($X + G3 = 2 + 8 \geq \tau, P2 + G3 = 9 + 8 \geq \tau$). (c) The extension-helpers attach ($O4 + V3 = 7 + 9 \geq \tau, V0 + V0 = 9 + 9 \geq \tau$).(d) The second extension-helper comes with a negative \textit{D} glue that causes targeted destabilization. The extension gadget and its helpers, along with the information block and its helpers are no longer bound to the path with sufficient strength($X + X + G1 + X + D = 2 + 2 + 8 + 2 - 7 \leq \tau$). (e) Once the unstable assemblies detach, a second half of the path piece attaches ($P2 + P2 = 9 + 9 \geq \tau$). (f) The final result is a one path-pixel extension of the path to the right of the direction the info block was walking.}
	\label{fig:extendright}
\end{figure}

\paragraph{Extend Right.}
The process for extending right is also very similar to extending forward. Just as before, the information block mush walk to the end of the path before it can be used for extension. The difference now, is the direction of the extension. The \textit{right} information block allows the right-extension gadget to attach and extend the path using the same mechanics as forward-extension. (Fig. \ref{fig:extendright}). The extension gadget reads the information block, extends the path in the given direction, and then detaches all but the newly extended path.

\begin{figure}[]
	\centering
	\begin{subfigure}[b]{.3\textwidth}
		\centering
		\includegraphics[scale=2.0]{Turning/walk_right_1.pdf}
		\caption{}
	\end{subfigure}
	$\quad$
	\begin{subfigure}[b]{.3\textwidth}
		\centering
		\includegraphics[scale=2.0]{Turning/walk_right_2.pdf}
		\caption{}
	\end{subfigure}
	$\quad$
	\begin{subfigure}[b]{.3\textwidth}
		\centering
		\includegraphics[scale=2.0]{Turning/walk_right_3.pdf}
		\caption{}
	\end{subfigure}
	\begin{subfigure}[b]{.3\textwidth}
		\centering
		\includegraphics[scale=2.0]{Turning/walk_right_4.pdf}
		\caption{}
	\end{subfigure}
	$\quad$
	\begin{subfigure}[b]{.3\textwidth}
		\centering
		\includegraphics[scale=2.0]{Turning/walk_right_5.pdf}
		\caption{}
	\end{subfigure}
	$\quad$
	\begin{subfigure}[b]{.3\textwidth}
		\centering
		\includegraphics[scale=2.0]{Turning/walk_right_6.pdf}
		\caption{}
	\end{subfigure}
	\caption{The right-walking gadget attaches to the path and the information block ($C + F + W = 4 + 1 + 5 \geq \tau$). (b) A new \textit{forward} information block attaches ($F + O3 + X = 1 + 7 + 2 \geq \tau$). The first walking-helper attaches to the new information block and the path ($G1 + Z3 = 8 + 9 \geq \tau$). (c) The negative interaction between the \textit{D} glues destabilizes the old information block, along with the two walking-helpers ($X + G1 + C + F + D = 2 + 8 + 4 + 1 - 7 \leq \tau$). (d) The old information block and walking-helpers break off, and the second walking-helper attaches ($Z4 + Z4 = 9 + 9 \geq \tau$). (e) Once the second walking-helper is attached, the walking gadget becomes unstable due to the negative \textit{Q} glues($F + O3 + E + Q = 1 + 7 + 5 - 4 \leq \tau$). (f) The walking gadget detaches and leaves the new information block with two walking-helpers attached to the other side of the left turn.}
	\label{fig:walkright}
\end{figure}

\paragraph{Walk Right.}
The walk-right process uses the same mechanics as walking forward or left, but again has slightly altered gadgets (Fig. \ref{fig:walkright}). The information block only allows the correct walking gadget to attach and begin the walking process. Once attached, the walking gadget allows a new information block (of the same type) to attach, while also detaching the the previous info block. After the previous information is removed, the walking gadget detaches as well, allowing the new info block to interact with other gadgets. Thus, the same information has traveled along the right-hand path turn. 


\section{Path Filling} \label{sec:filling}
This section covers the last step in shape construction. As shown in step 7 of the construction overview (Sec.~\ref{subsec:conover}), once the entire \emph{path} has been constructed, the space buffer around the path needs to be filled in. Here, we show the details of the filling process. After the buffer is filled in, the resultant assembly will be a scale 24 version of shape \emph{S}.

\begin{figure}[htp]
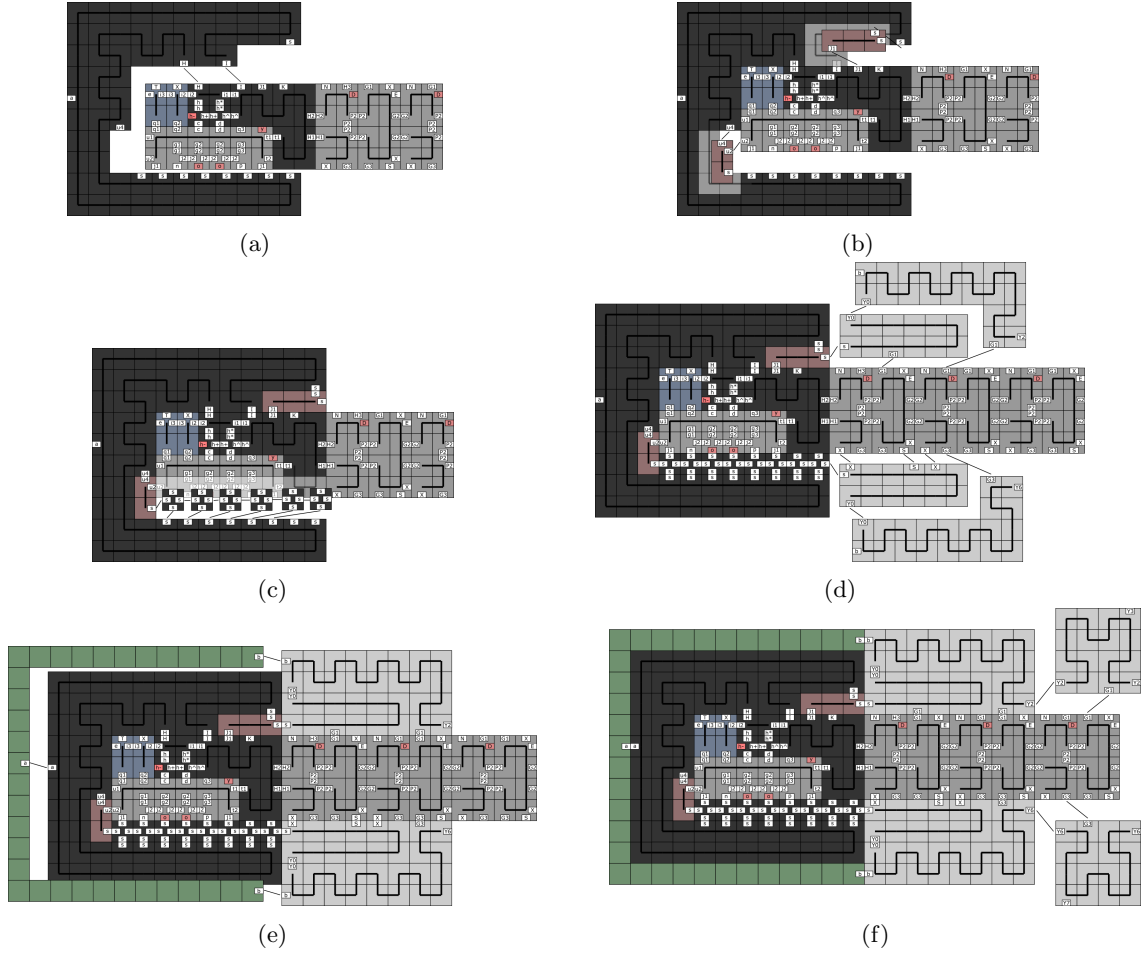

	\centering
	\begin{subfigure}[b]{.45\textwidth}
		\centering
		\includegraphics[scale=2.0]{Filling/fill_lines_1.pdf}
		\caption{}
	\end{subfigure}
	$\quad$
	\begin{subfigure}[b]{.45\textwidth}
		\centering
		\includegraphics[scale=2.0]{Filling/fill_lines_2.pdf}
		\caption{}
	\end{subfigure}
	$\quad$
	\begin{subfigure}[b]{.45\textwidth}
		\centering
		\includegraphics[scale=2.0]{Filling/fill_lines_3.pdf}
		\caption{}
	\end{subfigure}
	\begin{subfigure}[b]{.45\textwidth}
		\centering
		\includegraphics[scale=2.0]{Filling/fill_lines_4.pdf}
		\caption{}
	\end{subfigure}
	$\quad$
	\begin{subfigure}[b]{.45\textwidth}
		\centering
		\includegraphics[scale=2.0]{Filling/fill_lines_5.pdf}
		\caption{}
	\end{subfigure}
	$\quad$
	\begin{subfigure}[b]{.45\textwidth}
		\centering
		\includegraphics[scale=2.0]{Filling/fill_lines_6.pdf}
		\caption{}
	\end{subfigure}
	\caption{(a) The fill initiator gadget attaches to the tape ($H + I = 8 + 5 \geq \tau$). (b) Two fill helpers attack to the gadget/tape ($J1 + s = 8 + 8 \geq \tau, u4 + u2 = 8 + 8 \geq \tau$). (c) Single s-blocks attach ($s + s = 8 + 8 \geq \tau$).(d) The first topside fillers attach ($s + G1 = 8 + 8 \geq \tau, Y0 + G1 = 9 + 8 \geq \tau$), as well as the underside fillers ($X + S + X = 2 + 5 + 2 \geq \tau, Y0 + G3 = 9 + 8 \geq \tau$) (e) The filler cap attaches ($b + b + a = 3 + 3 + 4 \geq \tau$). (f) The topside and underside line filler blocks continue to attach ($Y2 + G1 = 9 + 8 \geq \tau, Y6 + G3 = 9 + 8 \geq \tau$).}
	\label{fig:linefill}
\end{figure}

\paragraph{Fill Lines.}
After the entire \emph{path} has been built, all previous tape sections will have been read/removed, save for one. The fill initiator gadget then attaches to the final tape section (Fig.~\ref{fig:linefill}), and begins the fill process. A series of cooperative attachments flood the sides (above and below) the path we've constructed. The initiator gadget, as well as the final tape section, remain to become the first pixel of the shape.

\begin{figure}[htp]
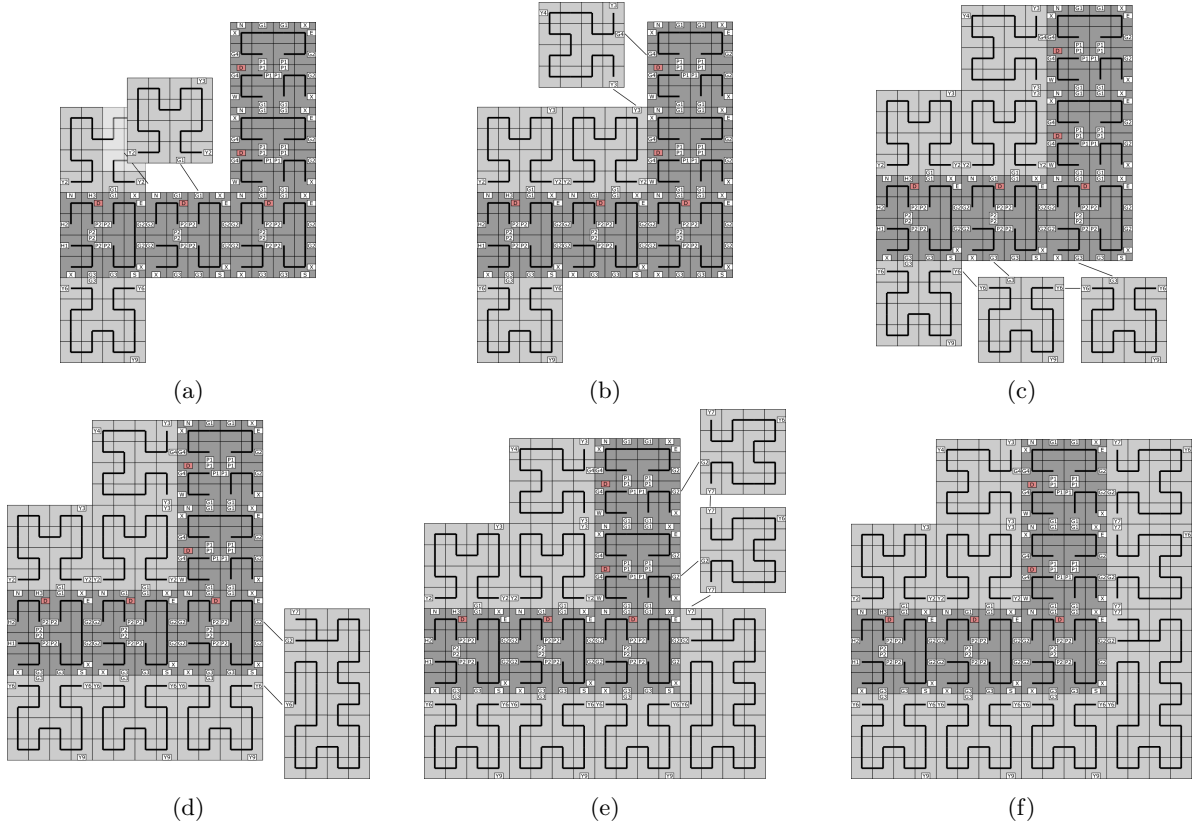

	\centering
	\begin{subfigure}[b]{.3\textwidth}
		\centering
		\includegraphics[scale=2.0]{Filling/fill_left_1.pdf}
		\caption{}
	\end{subfigure}
	$\quad$
	\begin{subfigure}[b]{.3\textwidth}
		\centering
		\includegraphics[scale=2.0]{Filling/fill_left_2.pdf}
		\caption{}
	\end{subfigure}
	$\quad$
	\begin{subfigure}[b]{.3\textwidth}
		\centering
		\includegraphics[scale=2.0]{Filling/fill_left_3.pdf}
		\caption{}
	\end{subfigure}
	\begin{subfigure}[b]{.3\textwidth}
		\centering
		\includegraphics[scale=2.0]{Filling/fill_left_4.pdf}
		\caption{}
	\end{subfigure}
	$\quad$
	\begin{subfigure}[b]{.3\textwidth}
		\centering
		\includegraphics[scale=2.0]{Filling/fill_left_5.pdf}
		\caption{}
	\end{subfigure}
	$\quad$
	\begin{subfigure}[b]{.3\textwidth}
		\centering
		\includegraphics[scale=2.0]{Filling/fill_left_6.pdf}
		\caption{}
	\end{subfigure}
	\caption{(a) The northern topside fillers attach until they encounter a left corner ($Y2 + G1 = 9 + 8 \geq \tau$). (b) The western topside filler attaches ($Y3 + G4 = 9 + 8 \geq \tau$). (c) The southern underside fillers attach until they reach a corner as well ($Y6 + G3 = 9 + 8 \geq \tau$).(d) The underside south-east filler attaches ($Y6 + G2 = 9 + 8 \leq \tau$) (e) The eastern underside fillers attach ($Y7 + G2 = 9 + 8 \geq \tau$). (f) A completely filled left turn.}
	\label{fig:leftfill}
\end{figure}

\paragraph{Fill Left.}
The filler blocks continue attaching until they encounter a corner (Fig.~\ref{fig:leftfill}). For left turns, the topside fillers encounter a concave corner, while the underside fillers encounter a convex corner. The design of the filler blocks allows them to simply transition from one block type to the next for concave corners. Convex corners, however, require a filler transition block to start filling in the new direction. Again, there are unique sets of filler blocks for filling along the topside and underside of the \emph{path}.

\begin{figure}[htp]
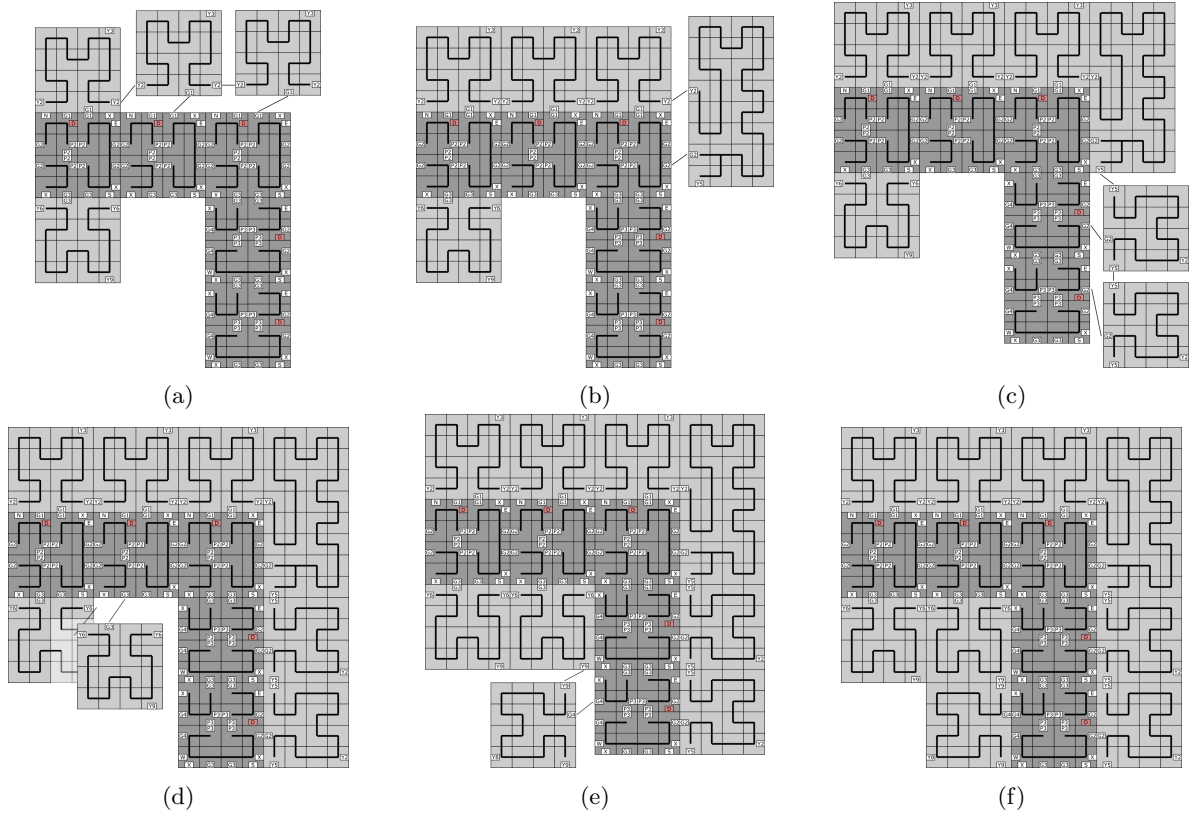

	\centering
	\begin{subfigure}[b]{.3\textwidth}
		\centering
		\includegraphics[scale=2.0]{Filling/fill_right_1.pdf}
		\caption{}
	\end{subfigure}
	$\quad$
	\begin{subfigure}[b]{.3\textwidth}
		\centering
		\includegraphics[scale=2.0]{Filling/fill_right_2.pdf}
		\caption{}
	\end{subfigure}
	$\quad$
	\begin{subfigure}[b]{.3\textwidth}
		\centering
		\includegraphics[scale=2.0]{Filling/fill_right_3.pdf}
		\caption{}
	\end{subfigure}
	\begin{subfigure}[b]{.3\textwidth}
		\centering
		\includegraphics[scale=2.0]{Filling/fill_right_4.pdf}
		\caption{}
	\end{subfigure}
	$\quad$
	\begin{subfigure}[b]{.3\textwidth}
		\centering
		\includegraphics[scale=2.0]{Filling/fill_right_5.pdf}
		\caption{}
	\end{subfigure}
	$\quad$
	\begin{subfigure}[b]{.3\textwidth}
		\centering
		\includegraphics[scale=2.0]{Filling/fill_right_6.pdf}
		\caption{}
	\end{subfigure}
	\caption{(a) The northern topside fillers attach until they encounter a right corner ($Y2 + G1 = 9 + 8 \geq \tau$). (b) The north-east topside filler attaches ($Y2 + G2 = 9 + 8 \geq \tau$). (c) The eastern topside fillers attach ($Y5 + G2 = 9 + 8 \geq \tau$).(d) The south underside fillers attach until they reach a corner as well ($Y6 + G2 = 9 + 8 \leq \tau$) (e) The western underside fillers attach ($Y9 + G4 = 9 + 8 \geq \tau$). (f) A completely filled right turn.}
	\label{fig:rightfill}
\end{figure}

\paragraph{Fill Right.}
The right-fill process is essentially a reflection of the left-turn process (Fig.~\ref{fig:rightfill}). Here, the topside fillers encounter the convex corner, and the underside fillers encounter the concave corner. Both filler types are designed to flood their respective sides of the \emph{path}. 

\section{Constant Scaled Shapes} \label{sec:shapes}
In this section, we formally state the results based on our construction.

\begin{theorem}\label{thm:main}
  For any finite connected shape $S$, there exists a 2HAM system $\Gamma=(T_{S},10)$ that uniquely produces $S$ (with $O(1)$ size bounded garbage) at a $O(1)$ scale factor, and $|T_{S}| = O( \frac{K(S)}{\log K(S)})$.
\end{theorem}

\begin{proof}
We show this by constructing a 2HAM system $\Gamma=(T_{S},10)$. One portion of $T_{S}$ consists of the tile types which assemble a higher base Kolmogorov-optimal description of \emph{S}, discussed in Section~\ref{sec:baseConversion}.  This portion of $T_{S}$ consists of $O( \frac{K(S)}{\log K(S)})$ tile types, as analyzed in Section~\ref{sec:baseConversion}. Another portion of $T_{S}$ consists of the tile types needed to assemble a fuel-efficient Turing machine, as described by \cite{SS2013FEC}, that performs a simple base conversion to binary using $O( \frac{K(S)}{\log K(S)})$ tile types, as analyzed in Section~\ref{sec:baseConversion}.  The next portion of $T_{S}$ consists of the tile types required to assemble another fuel-efficient Turing machine that finds and outputs the description of a path around the spanning tree of \emph{S}.  This Turing machine is of $O(1)$ size, and thus adds $O(1)$ tile types using the method from~\cite{SS2013FEC}. The final portion of $T_{S}$ consists of the tile types that construct the gadgets and assemblies shown in Section~\ref{sec:lineconstruction}. With the number of tile types used for computing the \emph{path} description and for our construction process being $O(1)$, our final tile complexity is $O(\frac{K(S)}{\log K(S)})$.

Now, consider assembly \emph{A} to be the fully constructed \emph{tape} assembly (detailed in Section~\ref{sec:lineconstruction}) encoded with \emph{path}-building instructions specific to \emph{S}. Also, consider assembly \emph{B} to be some \emph{terminal} assembly that has shape \emph{S} at a constant scale factor.

Note that $\Gamma$ follows the process detailed in Section~\ref{sec:lineconstruction}. This system was designed so that two assemblies are \emph{combinable} only if at least one of those assemblies is $O(1)$ size, and every \emph{breakable} assembly can only break into two subassemblies if one of those assemblies is $O(1)$ size. In our construction, the only non-constant size assemblies are \emph{A}, \emph{B}, or some intermediate assembly that consists of some portion of the \emph{tape}, and some partially assembled section of the final shape. Of these, \emph{B} is the only terminal assembly.

While \emph{A} and the intermediate assemblies continue engaging in a series of attachments and detachments, the \emph{tape} continues to get smaller and the \emph{path} continues to grow. The attachment and detachment of $O(1)$ size pieces with these assemblies will continue until we reach the terminal assembly \emph{B}, at which time \emph{A} will have been disassembled into smaller constant garbage. Therefore, we see that $A \rightarrow^{\Gamma} B$.
\end{proof}

\section{Lower Bound}\label{sec:lowerBound}

Here we present a brief argument for the lower bound of $\Omega(\frac{K(S)}{\log K(S)})$ on the tile types needed to assemble a scaling of a shape $S$.  This argument is essentially the same as what is presented in~\cite{AGKS05g,RotWin00,SolWin07}, and we refer the reader there for a more detailed explanation.

\begin{theorem}
The tile complexity in the 2HAM for self-assembling a scale-$c$ version of a shape $S$ at constant temperature and constant garbage is $\Omega(\frac{K(S)}{\log K(S)})$.
\end{theorem}
\begin{proof}
Note that a 2HAM system $\Gamma = (T,\tau=O(1))$ can be uniquely represented with a string of $O(|T|\log |T|)$ bits.  In particular, each tile may be encoded as a list of its 4 glues, and each glue may be represented by a $O(\log |T|)$-bit string taken from an indexing of the maximum possible $4|T|$ distinct glue types of the system.  The constant bounded temperature incurs an additional additive constant.  Given this representation, consider a 2HAM simulation program that inputs a 2HAM system, and outputs the positions of any uniquely produced scale-$c$ shape (with up to $O(1)$ garbage), if one exists.  This simulator, along with the $O(|T|\log |T|)$ bit encoding of a system $\Gamma$ which assembles $S$ at scale $c$, constitute a program which outputs the positions of $S$, and is thus lower bounded in bits by $K(S)$.  Therefore $K(S) \leq d|T|\log |T|$ for some constant $d$, implying $T = \Omega(\frac{K(S)}{\log K(S)})$.
\end{proof} 


\section{Extension to $\frac{K(S)}{\log K(S)}$}\label{sec:baseConversion}

The starting assembly for our shape construction algorithm is the \emph{tape} assembly from~\cite{SS2013FEC} with a binary string as its value.  For a binary string $A = a_0 \ldots a_{k-1}$, such an assembly can be constructed in a straightforward manner using $O(k)$ tile types (simply place a distinct tile for each position in the assembly, for example).  However, by using a base conversion trick, we can take advantage of the fact that each tile type is asymptotically capable of representing slightly more than 1 bit in order to build the string in $O(k/\log k)$ tile types.  To achieve this, first we consider the base-$b$ representation $B = b_0 \ldots b_{d-1}$ of the string $A$ for some higher base $b > 2$. Note that the number of digits of this string is $d \leq \lceil \frac{k}{\lfloor \log_2 b \rfloor}\rceil = O(\frac{k}{\log b})$.  We are able to assemble this shorter string (by brute force with distinct tile types at each position) with only $O(d)$ tile types.

Next, we consider a Turing machine which converts any base $b$ string into its equivalent base $2$ representation.  Such a Turing machine can be constructed using $O(b)$ transition rules.  Therefore, we can apply the result of~\cite{SS2013FEC} to run this Turing machine on the initial tape assembly representing string $B$ to obtain string $A$.  The cost of this construction in total is $O(d)$ tiles to construct the initial tape assembly, plus $O(b)$ tiles to implement the rules of the conversion Turing machine\footnote{The formal theorem statement of~\cite{SS2013FEC} cites the product of the states and symbols of the Turing machine as the tile type cost.  However, the actual cost is the number of transition rules, which is upper bounded by this product.}, for a total of $O(d+b)$ tiles.

Finally, we select $b=\lceil\frac{k}{\log k}\rceil = O(\frac{k}{\log k})$, which yields $d=O(\frac{k}{\log k - \log\log k}) = O(\frac{k}{\log k})$, implying that the entire tile cost of setting up the initial tape assembly representing binary string $B$ is $O(b+d) = O(\frac{k}{\log k})$ tile types.  In our case $k= O(K(S))$ where $K(S)$ denotes the Kolmogorov complexity of shape $S$ for some given universal Turing machine, and so we achieve our final tile complexity of $O( \frac{K(S)}{\log K(S)})$. 

\section{Conclusion} \label{sec:conclusion}
In this work, we considered the optimal shape building problem in the negative glue 2-handed assembly model, and provided a system of construction that allows the self-assembly of general shapes at scale 24. Shape construction has been studied in more powerful self-assembly models such as the staged RNA assembly model and the chemical reaction network-controlled tile assembly model. However, our result constitutes the first example of optimal general shape construction at constant scale in a \emph{passive} model of self-assembly where no outside experimenter intervention is required, and system monomers are state-less, static pieces which interact solely based on the attraction and repulsion of surface chemistry.

Our work opens up a number of directions for future work. We have not considered a runtime model for this construction, so analyzing and improving the \emph{running time} for constant-scaled shape self-assembly in the 2-handed assembly is one open direction.

Another is determining the lowest necessary temperature and glue strengths needed for $O(1)$ scale shape construction. We use temperature value 10 for the sake of clarity, and have not attempted to optimize this value. To help such optimization, we have included a table (Table~\ref{tbl:inequalities}) of the inequality specifications for every step in our construction process.

\begin{table}[H]
\centering\footnotesize
\begin{tabular}{ r r }
\hline
\multicolumn{2}{|c|}{Tape Overlay} \\
\hline
$t + $ \textit{h\^} $\geq \tau$ & $h + h^* + i \geq \tau$ \\
$i + q \geq \tau$ & $f + f^* + i \geq \tau$ \\
$e + w + q \geq \tau$ \\

\hline
\multicolumn{2}{|c|}{Tape Read} \\
\hline
$J + A \geq \tau$ & $n + T + F \geq \tau$ \\
$K + M \geq \tau$ & $J + F + F \geq \tau$ \\
$F + F + J + A + A + Q \geq \tau$ & $J + J + Q \geq \tau$ \\
$F + F + M + K + J + Q \geq \tau$ & $A + A + O \geq \tau$ \\
$F + F + M + n + T + F + Q < \tau$ \\

\hline
\multicolumn{2}{|c|}{Information Walk} \\
\hline
$F + F + J \geq \tau$ & $F + O + X \geq \tau$ \\
$J + Z + D \geq \tau$ & $F + O + J + D < \tau$ \\
$Z + Z + D \geq \tau$ & $Z + Z + J + D \geq \tau$ \\
$F + O + X + J + D \geq \tau$ & $J + A + A + F + F + D < \tau$ \\
$M + K + J + F + F + D \geq \tau$ & $J + A + J + F + F + D \geq \tau$ \\
$M + K + A + A + A + F + F + D \geq \tau$ & $J + O + J + F + F + D \geq \tau$ \\

\hline
\multicolumn{2}{|c|}{Path Extend} \\
\hline
$V + V + D \geq \tau$ & $H + X \geq \tau$ \\
$O + V + V + D \geq \tau$ & $V + O \geq \tau$ \\
$B + C + F + p \geq \tau$ & $H + P \geq \tau$ \\
$X + p + J + X + D < \tau$ & $P + P \geq \tau$ \\

\hline
\multicolumn{2}{|c|}{Tape Reduction} \\
\hline
$A + U \geq \tau$ & $u + u \geq \tau$ \\
$s + m + o \geq \tau$ & $s + s \geq \tau$ \\
$u + u + e + o + o < \tau$ \\ 

\hline
\multicolumn{2}{|c|}{Extend Left} \\
\hline
$V + V + D \geq \tau$ & $G + X \geq \tau$ \\
$B + C + L + X \geq \tau$ & $G + P \geq \tau$ \\
$X + X + G + X + D < \tau$ & $P + P \geq \tau$ \\

\hline
\multicolumn{2}{|c|}{Walk Left} \\
\hline
$F + O + X \geq \tau$ & $C + F + W \geq \tau$ \\
$F + O + W + Q < \tau$ & $Z + Z + Q \geq \tau$ \\
$C + F + X + G + D < \tau$ & $G + Z + D \geq \tau$ \\

\hline
\multicolumn{2}{|c|}{Extend Right} \\
\hline
$V + V + D \geq \tau$ & $V + O \geq \tau$ \\
$B + C + R + X \geq \tau$ & $P + G \geq \tau$ \\
$X + X + G + X + D < \tau$ & $X + G \geq \tau$ \\

\hline
\multicolumn{2}{|c|}{Walk Right} \\
\hline
$F + O + X \geq \tau$ & $C + F + E \geq \tau$ \\
$F + O + E + Q < \tau$ & $G + Z + D \geq \tau$ \\
$C + F + X + G + D < \tau$ & $Z + Z + Q \geq \tau$ \\

\hline
\multicolumn{2}{|c|}{Initial Fill} \\
\hline
$H + I \geq \tau$ & $u + u \geq \tau$ \\
$Y + G \geq \tau$ & $s + s \geq \tau$ \\
$b + b + a \geq \tau$ & $J + s \geq \tau$ \\
$s + X + S + X \geq \tau$ & $s + G \geq \tau$ \\

\hline
\multicolumn{2}{|c|}{Fill Forward, Left, Right} \\
\hline
$Y + G \geq \tau$ \\


\end{tabular}
\caption{Shown are the inequalities which must be satisfied for the construction gadgets shown in Section~\ref{sec:lineconstruction} to function in the way required to prove Theorem~\ref{thm:main}. All single glue labels  must have strength $< \tau$. Unless otherwise stated, in these inequalities, a glue label $G$ represents all glues $G*$ (e.g., G1, G2, etc.).}
\label{tbl:inequalities}
\end{table}

%

\bibliographystyle{amsplain}
\bibliography{tam}

\providecommand{\bysame}{\leavevmode\hbox to3em{\hrulefill}\thinspace}
\providecommand{\MR}{\relax\ifhmode\unskip\space\fi MR }
\providecommand{\MRhref}[2]{%
  \href{http://www.ams.org/mathscinet-getitem?mr=#1}{#2}
}
\providecommand{\href}[2]{#2}
\begin{thebibliography}{10}

\bibitem{AGKS05g}
Qi~Cheng, Gagan Aggarwal, Michael~H. Goldwasser, Ming-Yang Kao, Robert~T.
  Schweller, and Pablo~Moisset de~Espan\'{e}s, \emph{Complexities for
  generalized models of self-assembly}, SIAM Journal on Computing \textbf{34}
  (2005), 1493--1515.

\bibitem{DDFIRSS07}
Erik~D. Demaine, Martin~L. Demaine, S{\'a}ndor~P. Fekete, Mashhood Ishaque,
  Eynat Rafalin, Robert~T. Schweller, and Diane~L. Souvaine, \emph{Staged
  self-assembly: nanomanufacture of arbitrary shapes with ${O}(1)$ glues},
  Natural Computing \textbf{7} (2008), no.~3, 347--370.

\bibitem{RNAPods}
Erik~D. Demaine, Matthew~J. Patitz, Robert~T. Schweller, and Scott~M. Summers,
  \emph{Self-assembly of arbitrary shapes using rnase enzymes: Meeting the
  kolmogorov bound with small scale factor (extended abstract)}, Proceedings of
  the Twenty Eighth International Symposium on Theoretical Aspects of Computer
  Science (STACS 2011) (Dortmund, Germany), 2011.

\bibitem{DFS2015NGA}
ErikD. Demaine, SándorP. Fekete, Christian Scheffer, and Arne Schmidt,
  \emph{New geometric algorithms for fully connected staged self-assembly}, DNA
  Computing and Molecular Programming (Andrew Phillips and Peng Yin, eds.),
  Lecture Notes in Computer Science, vol. 9211, Springer International
  Publishing, 2015, pp.~104--116 (English).

\bibitem{Doty2013}
David Doty, Lila Kari, and Beno{\^i}t Masson, \emph{Negative interactions in
  irreversible self-assembly}, Algorithmica \textbf{66} (2013), no.~1,
  153--172.

\bibitem{MSS2012SWT}
Ján Maňuch, Ladislav Stacho, and Christine Stoll, \emph{Step-wise tile
  assembly with a constant number of tile types}, Natural Computing \textbf{11}
  (2012), no.~3, 535--550 (English).

\bibitem{PRS2016RMN}
Matthew~J. Patitz, Trent~A. Rogers, Robert Schweller, Scott~M. Summers, and
  Andrew Winslow, \emph{Resiliency to multiple nucleation in temperature-1
  self-assembly}, DNA Computing and Molecular Programming, Springer
  International Publishing, 2016.

\bibitem{rgTAM}
MatthewJ. Patitz, RobertT. Schweller, and ScottM. Summers, \emph{Exact shapes
  and turing universality at temperature 1 with a single negative glue}, DNA
  Computing and Molecular Programming, LNCS, vol. 6937, 2011, pp.~175--189.

\bibitem{REIF20111592}
John~H. Reif, Sudheer Sahu, and Peng Yin, \emph{Complexity of graph
  self-assembly in accretive systems and self-destructible systems},
  Theoretical Computer Science \textbf{412} (2011), no.~17, 1592 -- 1605.

\bibitem{Rothemund01022000}
Paul W.~K. Rothemund, \emph{Using lateral capillary forces to compute by
  self-assembly}, Proceedings of the National Academy of Sciences \textbf{97}
  (2000), no.~3, 984--989.

\bibitem{RotWin00}
Paul W.~K. Rothemund and Erik Winfree, \emph{The program-size complexity of
  self-assembled squares (extended abstract)}, Proc. of the 32nd ACM Sym. on
  Theory of Computing, STOC'00, 2000, pp.~459--468.

\bibitem{Schiefer2015}
Nicholas Schiefer and Erik Winfree, \emph{Universal computation and optimal
  construction in the chemical reaction network-controlled tile assembly
  model}, pp.~34--54, Springer International Publishing, Cham, 2015.

\bibitem{SS2013FEC}
Robert Schweller and Michael Sherman, \emph{Fuel efficient computation in
  passive self-assembly}, SODA 2013: Proceedings of the 24th Annual ACM-SIAM
  Symposium on Discrete Algorithms, SIAM, 2013, pp.~1513--1525.

\bibitem{SolWin07}
David Soloveichik and Erik Winfree, \emph{Complexity of self-assembled shapes},
  SIAM Journal on Computing \textbf{36} (2007), no.~6, 1544--1569.

\bibitem{Sum09}
Scott~M. Summers, \emph{Reducing tile complexity for the self-assembly of
  scaled shapes through temperature programming}, Algorithmica \textbf{63}
  (2012), no.~1, 117--136.

\end{thebibliography}

\newpage
\appendix

\section{Appendix}

\subsection{Special Case Gadgets} \label{app:specialgadgets}
\begin{figure}[H]
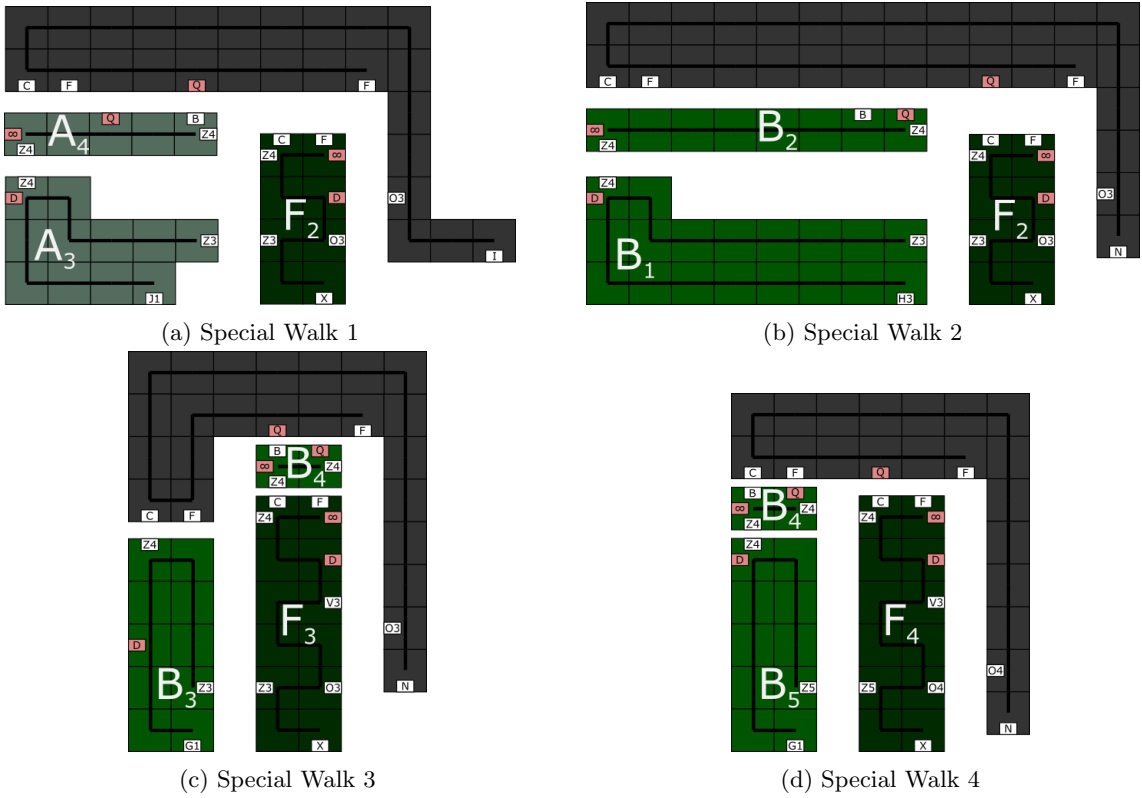


	\centering
	\begin{subfigure}[b]{.45\textwidth}
		\centering
		\includegraphics[scale=4.0]{ConstructionGadgets/Extras/special_walk_gadget_1.pdf}
		\caption{Special Walk 1}
	\end{subfigure}
	$\quad$
	\begin{subfigure}[b]{.45\textwidth}
		\centering
		\includegraphics[scale=4.0]{ConstructionGadgets/Extras/special_walk_gadget_2.pdf}
		\caption{Special Walk 2}
	\end{subfigure}
    $\quad$
	\begin{subfigure}[b]{.45\textwidth}
		\centering
		\includegraphics[scale=4.0]{ConstructionGadgets/Extras/special_walk_gadget_3.pdf}
		\caption{Special Walk 3}
	\end{subfigure}
    $\quad$
	\begin{subfigure}[b]{.45\textwidth}
		\centering
		\includegraphics[scale=4.0]{ConstructionGadgets/Extras/special_walk_gadget_4.pdf}
		\caption{Special Walk 4}
	\end{subfigure}
	\caption{These are the special-case walking gadgets. Like the standard gadget, each special-case also has specific versions for the different instructions that are carried by the information blocks. While slightly different, all walking gadgets utilize the same mechanics shown in Section~\ref{sec:lineconstruction}. (a-d) The walking gadget (dark grey) has versions specific to the different information types. The \emph{F} glues indicate that these walkers will be walking an \emph{F}-type instruction down the path. The negative \emph{Q} glue allows the gadget to be detached once the information block has walked one step. The dark green assembly is the corresponding information block. The two other helpers (olive drab/ light green) remove the previous info block/helpers and the walking gadget.}
	\label{fig:specialwalkgadgets}
\end{figure}

\noindent \textbf{Special Case Walking Gadgets}
All walking gadgets are mechanically identical to the one described in Section~\ref{sec:lineconstruction}. The glues are unique to allow special versions of information blocks and helpers to attach, but the process is the same. The walkers attach to a previously placed information block, allow the attachment of a new information block (with helpers), and detach themselves and the previous information block. The changes here are due to \emph{what} they are walking on, be it the \emph{tape} or the \emph{path}.

\textbf{Special Walk 1 \& 2}
The first special walking gadget (Fig~\ref{fig:specialwalkgadgets}a) is used for all subsequent steps along the \emph{tape}, after the standard walking gadget (Sec~\ref{sec:lineconstruction}) has executed the information block's initial step. The second special walking gadget (Fig~\ref{fig:specialwalkgadgets}b) allows the information block to transition from walking on the \emph{tape}, to walking on the \emph{path}. This gadget is required because single \emph{tape} and \emph{path} sections differ in size and glue types.

\textbf{Special Walk 3 \& 4}
Once an information block has transitioned from the \emph{tape} to the \emph{path}, the third special walking gadget (Fig~\ref{fig:specialwalkgadgets}c) allows it to take an initial step on the \emph{path}. The fourth special walking gadget (Fig~\ref{fig:specialwalkgadgets}d) allows the information block to continue walking along the \emph{path}. These gadgets are required because a single \emph{path} section is shorter than a single \emph{tape} section, which these gadgets account for.

\begin{figure}[]
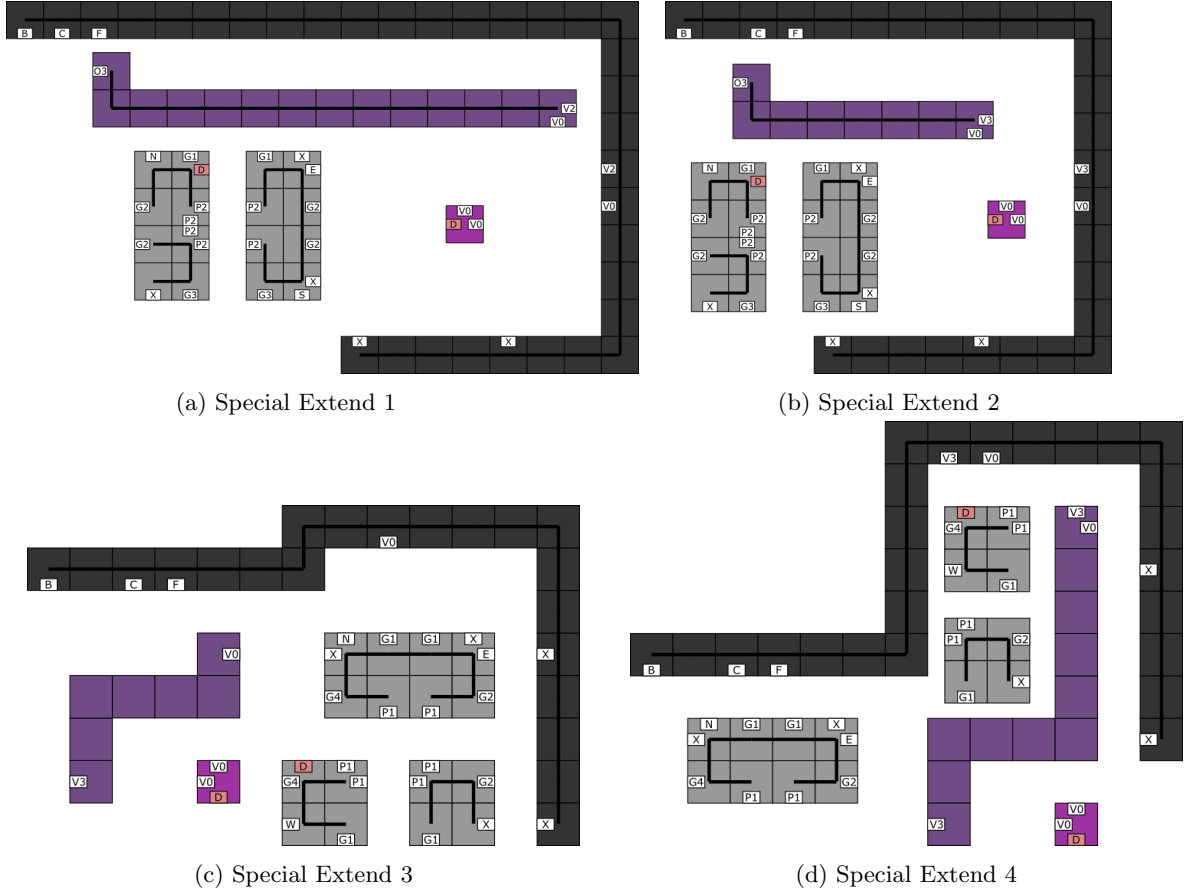

	\centering
	\begin{subfigure}[b]{.45\textwidth}
		\centering
		\includegraphics[scale=3.5]{ConstructionGadgets/Extras/special_extend_gadget_1.pdf}
		\caption{Special Extend 1}
	\end{subfigure}
	$\quad$
	\begin{subfigure}[b]{.45\textwidth}
		\centering
		\includegraphics[scale=3.5]{ConstructionGadgets/Extras/special_extend_gadget_2.pdf}
		\caption{Special Extend 2}
	\end{subfigure}
    $\quad$
	\begin{subfigure}[b]{.45\textwidth}
		\centering
		\includegraphics[scale=4.0]{ConstructionGadgets/Extras/special_extend_gadget_3.pdf}
		\caption{Special Extend 3}
	\end{subfigure}
    $\quad$
	\begin{subfigure}[b]{.45\textwidth}
		\centering
		\includegraphics[scale=4.0]{ConstructionGadgets/Extras/special_extend_gadget_4.pdf}
		\caption{Special Extend 4}
	\end{subfigure}
	\caption{These are the special-case extension gadgets. While slightly different, all extension gadgets utilize the same mechanics shown in Section~\ref{sec:lineconstruction}. (a-d) The \emph{B, C, F, X} glues on the extension gadgets (dark grey) are used for gadget attachment. The second \emph{X} glue allows the first portions of the new path piece (light grey) to attach. The \emph{V} glues are utilized by the extension helpers (purple) in order to detach the gadget.}
	\label{fig:specialextendgadgets}
\end{figure} 


\noindent \textbf{Special Case Extension Gadgets}
As with the walking gadgets, all extension gadgets are mechanically the same as the gadget introduced in Section~\ref{sec:lineconstruction}. While their function of these gadgets are the same (to extend the \emph{path} by one section), the primary difference is in their geometry.

\textbf{Special Extend 1 \& 2}
Two special gadgets (Figure~\ref{fig:specialextendgadgets}a-b) are required to extend the path after the standard extension gadget (Sec~\ref{sec:lineconstruction}) builds the first \emph{path} portion. The special extend 1 gadget (Fig~\ref{fig:specialextendgadgets}a) is designed to build the second \emph{path} portion. Once the second \emph{path} portion has been built, the special extend 2 gadget (Fig~\ref{fig:specialextendgadgets}b) carries out all remaining forward extensions, with the exception of two special cases after a left turn.

\textbf{Special Extend 3 \& 4}
After a left extend (Sec~\ref{sec:turning}, two more special gadgets (Fig~\ref{fig:specialextendgadgets}c-d) are required to extend the \emph{path} to a sufficient length that allows the standard walking/extending gadgets to be used. The first two of these forward extensions require special case gadgets. The first of those gadgets (Fig~\ref{fig:specialextendgadgets}c) extends the \emph{path} in the new direction after the turn. The second forward extension after a left turn requires a special gadget as well (Fig~\ref{fig:specialextendgadgets}d). This gadget builds an additional \emph{path} portion in the direction of the turn.

\end{document}